\documentclass[11pt]{article}
\usepackage{times}
\usepackage{latexsym}
\usepackage{graphicx, epsfig}
\usepackage{amsmath, amsthm, amsfonts}
\usepackage{algorithm}
\usepackage{supertabular}
\usepackage{subfigure}
\usepackage[noend]{algorithmic}
\usepackage{fullpage}
\usepackage{smallhead}

\newtheorem{theorem}{Theorem}
\newtheorem{definition}{Definition}

\newtheorem{lemma}{Lemma}

\newcommand{\BfPara}[1]{\noindent{\bf #1}.}
\newenvironment{junk}[1]{}

\title{{\bf Preference Games and Personalized Equilibria, with
Applications to Fractional BGP}}

\author{Laura J. Poplawski\\
College of Computer Science\\
and Information Science\\
Northeastern University\\
ljp@ccs.neu.edu
 \and Rajmohan Rajaraman \\
College of Computer Science\\
and Information Science\\
Northeastern University\\
rraj@ccs.neu.edu
\and Ravi Sundaram \\
College of Computer Science\\
and Information Science\\
Northeastern University\\
koods@ccs.neu.edu
\and Shang-Hua Teng\\
Department Computer Science\\
Boston University \\
\\
steng@cs.bu.edu}

\begin{document}

\newcommand{\gamename}{preference game}
\newcommand{\gamenameCaps}{Preference Games}
\newcommand{\gamenamePlural}{preference games}
\newcommand{\gamenamePluralInitCap}{Preference games}

\maketitle
\begin{abstract}

We study the complexity of computing equilibria in two classes of
network games based on flows - fractional BGP (Border Gateway
Protocol) games and fractional BBC (Bounded Budget Connection) games.
BGP is the glue that holds the Internet together and hence its
stability, i.e. the equilibria of fractional BGP games
\cite{HaxellWilfong}, is a matter of practical importance. BBC games
\cite{OurPODC08} follow in the tradition of the large body of work on
network formation games and capture a variety of applications ranging
from social networks and overlay networks to peer-to-peer networks.

The central result of this paper is that there are no fully polynomial-time approximation
schemes (unless \textbf{PPAD} is in \textbf{FP}) for computing equilibria in both fractional BGP games and 
fractional BBC games. We obtain this result by proving the hardness for a new and surprisingly simple 
game, the {\em \gamename}, which is reducible to both fractional BGP and BBC games. 

We define a new flow-based notion of equilibrium for matrix games -- {\em personalized equilibria} -- which generalizes both fractional BBC games and fractional BGP games. We prove not just the
existence, but the existence of \emph{rational} personalized equilibria for all matrix games, which implies the existence
of rational equilibria for fractional BGP and BBC games. In particular, this provides an alternative
proof and strengthening of the main result in \cite{HaxellWilfong}. For $k$-player matrix games, 
where $k = 2$, we 
provide a combinatorial characterization leading to a polynomial-time algorithm for
computing all personalized equilibria. For $k \geq 5$, we prove that personalized equilibria are 
\textbf{PPAD}-hard to approximate in fully polynomial time. We believe that the concept of 
personalized equilibria has potential for real-world significance.

\junk{
The central result of this paper is that there are no fully polynomial-time approximation
schemes (unless PPAD is in FP) for computing equilibria in both fractional BGP games and 
fractional BBC games. 

The Border Gateway Protocol (BGP), used to exchange routing information, 
is the glue that holds the Internet together. The stability of BGP 
is a matter of practical importance and has been studied extensively. Recently, 
a fractional model was introduced and studied in the context of game theory \cite{} 
--- the fractional BGP game. 

There is a considerable body of work on network formation games. Bounded 
Budget Games (BBC) games, introduced and studied in \cite{} capture a variety 
of applications ranging from social networks and overlay networks to peer-to-peer 
networks from the standpoint of game theory. The object of our study in this paper
is the fractional version --- the fractional BBC game.

Generalizing both fractional BBC games and fractional BGP games we define a new flow-based
notion of equilibrium for matrix games -- personalized equilibria. We prove not just the
existence, but the existence of \emph{rational} personalized equilibria, which implies the existence
of rational equilibria for fractional BGP and BBC games. In particular this provides an alternative
proof and strengthening of the main result in \cite{}. For k-player matrix games, where k = 2, we 
provide a combinatorial characterization based on which we develop a polynomial-time algorithm for
computing the personalized equilibrium. We
}

\junk{
We study the complexity of computing equilibria in two classes of
network games based on fractional flows.  We first consider a
fractional model of the Border Gateway Protocol (BGP), the interdomain
routing protocol used to exchange routing information in the Internet
today.  The selection of routes in BGP can be modeled as a game, whose
pure Nash equilibria correspond to a set of stable routes for given
route preferences.  It has been shown that there exist instances of
this game for which pure Nash equilibria do not exist, thus implying
that for certain sets of route preferences BGP may not converge.  In
recent work, Haxell and Wilfong (SODA 2008) introduced a fractional
model for BGP and showed that when fractional routes are allowed,
equilibria always exist.  We prove that the problem of computing an
equilibrium in the fractional BGP game is \textbf{PPAD}-hard.  We next
study a fractional version of Bounded Budget Connection (BBC) games
introduced by Laoutaris {\em et al} (PODC 2008).  We prove that the
problem of computing an equilibrium of the fractional BBC game is also
\textbf{PPAD}-hard.  Furthermore, we show that both of these
computational problems have no fully-polynomial-time approximation
scheme unless \textbf{PPAD} is in \textbf{FP}.

Our complexity results are built on a new understanding of the
mathematical structures of these games.  In particular, we give an
alternative proof of the existence of stable paths in the fractional
BGP game, and prove that rational solutions always exist for both
fractional BGP and fractional BBC games.  Our study has also led us to
define a new flow-based notion of equilibria for matrix games, which
we call {\em personalized equilibria}.  We give a combinatorial
characterization of two-player personalized equilibria, building on
which we develop a polynomial-time algorithm.  In contrast, we prove
that for $k \geq 4$, personalized equilibria are \textbf{PPAD}-hard to
approximate in fully polynomial time.
}
\end{abstract}

\thispagestyle{empty}
\newpage
\setcounter{page}{1}

\section{Introduction}
This paper concerns two classes of games on networks involving
fractional flows --- fractional BGP games and fractional BBC
games. These games model important practical systems such as the
Internet and social networks.  The stable operating points of these
systems have real-world significance and hence there is interest in
their pure Nash equilibria. In order to understand the structure and
computational complexity of these equilibria, we define two new
concepts --- personalized equilibria for matrix games and
\gamenamePlural { }--- which are of independent interest. Below we
briefly describe and motivate each of the four different kinds of
games.

\textbf{Fractional BGP games.} The Border Gateway Protocol is the core routing protocol of the Internet. 
\junk{It does not solve the shortest paths problem because Autonomous Systems are strategic agents that exchange
information with a view to optimizing their preferred paths to the destination \cite{GriffinShepherdWilfong}.}
BGP can be viewed as a distributed mechanism for solving the stable paths problem \cite{GriffinShepherdWilfong}. 
In this paper,
we refer to the fractional version of the stable paths problem introduced in \cite{HaxellWilfong} as the
\emph{fractional BGP game}. Intuitively, the fractional BGP game is a game played between Autonomous Systems that
assign fractional capacities to the different paths leading to the destination in such a way that they maximize
their utility without violating the capacity constraints of downstream nodes. Clearly, the equilibria
of this game have significant implications for the stability of the Internet. \junk{Their structure and computational
complexity impact the operation and administration of the current Internet as well as the design principles 
and management policies for the next generation Internet.}

\textbf{Fractional BBC games.} Consider a social network where people
have to spend time and (cognitive) resources to build connections to
people. This situation naturally lends itself to being modeled by the
Bounded Budget Connection game \cite{OurPODC08, LaoutarisPRST}. In a BBC game,
strategic nodes acting under a cost budget form connections with a
view to optimizing their proximity to nodes of interest, which in turn
depends on other nodes' strategic actions.  BBC games belong to the
much studied class of network formation games. These games have
applications to a variety of problems ranging from ``how to monetize a
social network'' to ``how to structure incentives in a peer-to-peer
network to reduce congestion.''  Fractional BBC games were defined
in~\cite{LaoutarisPRST}, which left unresolved the complexity of finding
their Nash equilibria.  

\textbf{Personalized equilibria for matrix games - a generalization.}
Imagine a business selling outfits consisting of a pant (solid or striped) and a shirt
(cotton or wool).  The manager of one location decides on the ratio of
striped to solid pants while the manager at the other decides on the
ratio of cotton to wool shirts. Each manager is given the same number
of shirts and pants (in the proportions decided) and has to assemble
and sell the outfits at her own location in such a way as to
maximize her individual profits. Personalized equilibria for matrix
games capture exactly this situation: each player chooses a
distribution over her own actions, but then each player 
independently customizes the matching of her own actions to the
actions of other players in such a way as to maximize individual payoff. The concept
of personalized equilibria for matrix games generalizes both
fractional BGP and fractional BBC games.

\textbf{\gamenamePluralInitCap { }- a specialization.} It is New
Year's Eve. You and each of your friends is hosting a party. Each of
you has a preference order over the others' parties and has to
determine the fraction of the evening that you will spend at each
party.  Naturally, one cannot spend more time at a party
than the person hosting that particular party.  Your optimal action --
how long to host your party and which other parties to attend for how
long -- depends on your preference and other players' actions.  Such
\gamenamePlural\ arise whenever each player has a preference among her
actions and her distribution over her actions is somewhat constrained
by others' distributions.  \gamenamePluralInitCap\ are reducible in
polynomial-time to both fractional BGP and BBC games.

\junk{
The \gamename 
problem is to find the optimal way to spend the evening.  Preference 

Though
simple to describe, \gamenamePlural { }are hard to solve and are
reducible in polynomial-time to both fractional BGP and BBC games.
}
\junk{, i.e., what fraction of time to spend at each party. This
situation is exactly captured by the class of games, termed party
games, that we define in this paper. Party games are very simple to
describe, nevertheless, as we will show in this paper finding their
equilibria is not easy. Further, party games are reducible in
polynomial-time to both fractional BGP and fractional BBC games. Thus,
the hardness of finding the equilibria of party games implies the
hardness of finding equilibria of both fractional BGP and fractional
BBC games.  }

\subsection{Our Contributions}

Our paper centers on the study of two classes of flow-based network
games --- fractional BGP games and fractional BBC games, formally
defined in Section~\ref{sec:definitions}.  We address
the following two questions: \junk{1. Does a Nash equilibrium always exist
and if so what does the set of equilibria look like? 2. From a
computational standpoint, how difficult is it to find a Nash
equilibrium?}

\begin{enumerate}
\item Does a Nash equilibrium always exist and if so what does the set of equilibria look like?
\item From a computational standpoint how difficult is it to find a Nash equilibrium?
\end{enumerate} 

To answer the first question, we define a new flow-based notion of
equilibrium for matrix games -- {\em personalized
equilibria}. Personalized equilibria for matrix games constitute
a novel and useful generalization of the concept of Nash equilibria
for both fractional BGP as well as fractional BBC games. By
employing this generalization, we are able to show the following
results for multi-player matrix games.
\vspace{-2mm}
\begin{itemize}
\item We show that the set of personalized equilibria for any
multi-player matrix game is always nonempty and contains a rational
point, though it may be nonconvex (Section~\ref{sec:personalized}).
\end{itemize}
\vspace{-2mm}
It follows that a rational equilibrium always exists for both fractional BGP and BBC games.
We provide an alternate characterization for BGP games that enables a simpler existence proof and strengthens the 
main result in \cite{HaxellWilfong}.
We expect that personalized equilibria will be applicable elsewhere,
since they capture real-world situations in which players have the
opportunity to customize their use of others' actions.

To answer the second question, we create a new combinatorial
$k$-player abstract game -- the {\em \gamename} (see
Section~\ref{sec:hardness}).  \gamenamePluralInitCap { }are extremely
elementary games that are a simultaneous simplification of both
fractional BGP and fractional BBC games. By employing this simplified
abstraction, we are able to obtain the following by reduction from a
Brouwer fixed point
problem~\cite{ChenDengTeng,ChenDengTengJACM,DaskalakisGoldbergPapadimitriou}.
\vspace{-2mm}
\begin{itemize}
\item There are no fully polynomial-time approximation schemes (unless
\textbf{PPAD} is in \textbf{FP}) for computing equilibria in
\gamenamePlural\ (Section~\ref{sec:hardness}).
\end{itemize}
\vspace{-2mm} It follows that there are no fully polynomial-time
approximation schemes (unless \textbf{PPAD} is in \textbf{FP}) for
computing equilibria in fractional BGP and fractional BBC games, as
well as personalized equilibria in multi-player matrix games. Our
result for fractional BGP games settles a question left open in
\cite{HaxellWilfong}, while our result for fractional BBC games
settles an open question from \cite{OurPODC08}.  Finally, we study the
complexity of personalized equilbria in $k$-player matrix games, for
fixed $k$.
\begin{itemize}
\vspace{-2mm}
\item For $k = 2$, we provide a combinatorial characterization which
implies a polynomial-time algorithm for computing the personalized
equilibrium (Section~\ref{sec:personalized}).
\vspace{-2mm}
\item 
For $k \ge 4$, it is \textbf{PPAD}-hard to find personalized
equilibria.  Furthermore, for $k \ge 5$, there is no fully polynomial time
approximation scheme (unless \textbf{PPAD} is in \textbf{FP}) for
finding personalized equilibria (Section~\ref{sec:personalized}).
\end{itemize}
\vspace{-2mm}

\junk{The reduction 
involves the creation of gadgets that capture the actions of gates 
in a general Boolean circuit. 
It is our expectation that on account of its simplicity the party game, as well as the gadgets involved in
the reduction, will
find use in other \textbf{PPAD}-hardness reductions.

The following is a schematic of the main results and structure of this paper.

}

\subsection{Related Work}

Nash equilibrium \cite{Nash50, Nash51} is arguably the most influential
solution concept in game theory. Decades after Nash,
Papadimitriou defined a complexity class \textbf{PPAD} \cite{Papadimitriou94} to characterize 
proofs that rely on parity arguments. Recently, an exciting breakthrough made in 
\cite{DaskalakisGoldbergPapadimitriou} and strengthened in \cite{ChenDengTeng} showed the hardness of 
approximating Nash equilibria.
Since then, there has been a flurry of work on the complexity of finding
equilibria in a variety of games and markets
\cite{NisanRoughgardenTardosVazirani}.  All of our hardnesss proofs
build on the framework established
by~\cite{ChenDengTeng,ChenDengTengJACM,DaskalakisGoldbergPapadimitriou}
and heavily use their techniques.
\junk{Following 
in this tradition, our
work here centers on two classes of network games based on fractional flows.}

BGP has been the focus of
much attention since its inception \cite{RehkterLi, Stewart}. The integral stable paths
problem was introduced \cite{GriffinShepherdWilfong} to explain the 
nonconvergence
of BGP \cite{VaradhanGovindanEstrin}. The fractional relaxation of the stable paths problem, 
or the fractional BGP game, was defined in 
\cite{HaxellWilfong} where they proved the existence of an equilibrium but left open the complexity
of finding it. 
\cite{Kintali} gives a distributed algorithm for finding 
an $\epsilon$-approximation for the fractional BGP game that is guaranteed 
to converge, although no bounds are given on the time-to-convergence (the main result in
this paper implies a polynomial upper bound is unlikely). Other related works include a multicommodity version 
\cite{Papadimitriou01, MarkakisSaberi} and mechanism design \cite{FeigenbaumPapadimitriouSamiShenker}.

The BBC game, introduced in \cite{OurPODC08, LaoutarisPRST}, builds on a large body
of work in network formation games \cite{JacksonWolinsky, BalaGoyal}.
A direct precursor to BBC games was introduced in
\cite{FabrikantLuthraManevaPapadimitriouShenker} .
\cite{FabrikantLuthraManevaPapadimitriouShenker}, which together with
subsequent works \cite{AlbersEiltsEvenDarMansourRoditty,
DemaineHajiaghaviMahini}, focuses on obtaining price of anarchy
results \cite{KoutsoupiasPapadimitriou}.  In \cite{OurPODC08}, it is
shown that it is \textbf{NP}-hard \cite{GareyJohnson} to determine
whether an equilibrium exists in integral BBC games. Fractional BBC
games were also introduced in \cite{LaoutarisPRST}, but the problem of
finding an equilibrium was left open. Other related works include a
stochastic small-world version \cite{EvenDarKearns} and use of
contracts \cite{AnshelevichShepherdWilfong, JohariMannorTsitsiklis}.

\junk{
The fractional stable paths problem that we view as a fractional BGP game was defined by 
Haxell and Wilfong \cite{HaxellWilfong}. However, the integral stable paths problem 
was introduced much earlier, by Griffin, Shepherd, and Wilfong \cite{GriffinShepherdWilfong}, 
in order to explain the nonconvergence of BGP, as observed by Varadhan, Govindan, and Estrin 
\cite{VaradhanGovindanEstrin}. Kintali gives a distributed algorithm for finding 
an $\epsilon$-approximation for the fractional stable paths problem that is guaranteed 
to converge, although no bounds are given on the time-to-convergence \cite{Kintali}. 
Another game-theoretic model of BGP is  introduced in \cite{Papadimitriou01} and studied 
further in \cite{MarkakisSaberi}. In this model, we are given a multi-commodity flow in a 
network, and each node's payoff is based on the flow that starts or ends at that node. BGP as a 
protocol was introduced by Rehkter and Li in 1995. \cite{RehkterLi}. 

Also related is the work on designing strategy-proof mechanisms for
BGP~\cite{FeigenbaumPapadimitriouSamiShenker} as well as the recent work on
strategic network formation through AS-level
contracts~\cite{AnshelevichShepherdWilfong}. \cite{JohariMannorTsitsiklis} consider a 
contracts-based model of network formation where links do not have predefined
costs but are subject to negotiation and nodes attempt to minimize incoming
traffic by obtaining compensation in return.

The BBC game was introduced in \cite{OurPODC08}, building on a large body of work in network 
formation games, or games in which nodes may purchase links to other nodes in order to minimize 
some cost function or maximize some utility function.  In \cite{OurPODC08}, the authors show that 
pure Nash equilibria do not always exist in integral BBC and that it is NP hard to determine 
whether an equilibrium exists. However, the also consider a uniform case for which they give a 
class of Nash equilibria and fairly tight price of anarchy bounds, which are bounds on the ratio 
of the worst-case Nash equilibrium to the social optimum cost \cite{KoutsoupiasPapadimitriou}. 
\cite{FabrikantLuthraManevaPapadimitriouShenker} is the closest predecessor to BBC games: in their 
model, nodes have a utility function that includes the cost of building links as well as the distance 
to the other nodes. They give price of anarchy bounds in their model. 
\cite{AlbersEiltsEvenDarMansourRoditty} and ~\cite{DemaineHajiaghaviMahini} give more 
results for the same model. \cite{HaleviMansour} extends this model to the case where each node 
is only interested in connecting to a subset of the other nodes. \cite{EvenDarKearns} also imposes 
a cost for purchasing links, but considers a stochastic model and associated small-world effects. 

Earlier network creation games were studied by Jackson and Wolinsky \cite{JacksonWolinsky}, 
who study pairwise stability, and Bala and Goyal \cite{BalaGoyal}, who consider a model where 
nodes incur cost based on the number of incoming links. 

}

\newcommand{\prefer}{\ge}
\newcommand{\eps}{\varepsilon}
\section{Definitions} 
\label{sec:definitions}
In this section, we define fractional BGP and fractional BBC
games.  Our definitions lead to the existence of Nash equilibria for
these games using standard fixed-point techniques.  We defer the
formal proofs of existence of equilibria, however, to
Section~\ref{sec:personalized}, where we establish the existence of a
more general class of equilibria that includes the equilibria for both
fractional BGP and BBC games.

\subsection{Fractional BGP}
\label{sec:definitions.BGP}
The fractional BGP game is based on a new model
of~\cite{HaxellWilfong}, introducing the notion of a fractional stable
paths solution in the context of BGP.  We first present the model
of~\cite{HaxellWilfong}, and then define the fractional BGP game,
whose Nash equilibria are equivalent to fractional stable paths
solutions.

Let $G$ be a graph with a distinguished node $d$, called the {\em
destination}.  Each node $v
\neq d$ has a list $\pi(v)$ of simple paths from $v$ to $d$ and a
preference relation\footnote{A preference relation is a binary
relation that is transitive and complete.}  $\prefer_v$ among the
paths in $\pi(v)$.  For paths $P$ and $P'$ in $\pi(v)$, $P \prefer_v
P'$ indicates that $v$ prefers $P$ at least as much as $P'$.  We say
that $P >_v P'$ if $P \prefer_v P'$ is true but $P'\prefer_v P$ is not
true.  When it is clear from context that we are talking about the
preferences for node $v$, we will write $P \prefer P'$ instead of $P
\prefer_v P'$. For a path $S$, we also define $\pi(v, S)$ to be the
set of paths in $\pi(v)$ that have $S$ as a suffix. A \emph{proper
suffix} $S$ of $P$ is a suffix of $P$ such that $S \neq P$ and $S \neq
\emptyset$.

A {\em feasible fractional paths solution}\/ is a set $w = \{w_v: v
\neq d\}$ of assignments $w_v: \pi(v) \rightarrow [0,1]$ satisfying
the following: 
\begin{enumerate}
\item {\bf Unity condition}: for each node $v$, $\sum_{P
\in
\pi(v)} w_v(P) \le 1$
\item {\bf Tree condition}: for each node $v$, and each
path $S$ with start node $u$, $\sum_{P \in
\pi(v, S)} w_v(P) \le w_u(S)$.  
\end{enumerate}
In other words, a feasible solution is one in
which each node chooses at most 1 unit of flow to $d$ such that no
suffix is filled by more than the amount of flow placed on that
suffix by its starting node. A feasible solution $w$ is {\em
stable}\/ if for any node $v$ and path $Q$ starting at $v$, one of the
following holds: 
\begin{itemize}
\item[{\bf (S1)}] $\sum_{P \in \pi(v)} w_v(P) = 1$,
and for each $P$ in $\pi(v)$ with $w_v(P) > 0$, $P \prefer_v
Q$; or 
\item[{\bf (S2)}] There exists a proper suffix $S$ of $Q$ such
that $\sum_{P \in \pi(v,S)} w_v(P) = w_u(S)$, where $u$ is the start
node of $S$, and for each $P \in \pi(v,S)$ with $w_v(P) > 0$, 
$P\prefer_v Q$.
\end{itemize}
  In other words, in a stable solution: if node $v$ has
not fully chosen paths that it prefers at least as much as $Q$, then
it has completely filled path $Q$ by filling some suffix with paths
it prefers at least as much as $Q$.

We now define the fractional BGP game.  For convenience, let $w_{-v}$
denote $\{w_u: u \neq d, v\}$.  Given assignments $w_v$, $w'_v$,
and $w_{-v}$ such that $(w_v, w_{-v})$ and $(w'_v, w_{-v})$ are both
feasible, we say $w_v$ is {\em lexicographically at least}\/
$w'_v$ (implied: with respect to $w_{-v}$)\/ if the following holds for
every path $P$ in $\pi(v)$: $\sum_{P' \prefer P} w_v(P') \ge
\sum_{P' \prefer P} w'_v(P')$.  We say that $w_v$ is {\em
lexicographically maximal} (implied: with respect to $w_{-v}$)\/ if $(w_v,
w_{-v})$ is feasible and $w_v$ is lexicographically at least every
assignment $w'_v$ such that $(w'_v, w_{-v})$ is feasible.

In the \emph{fractional BGP Game}, a strategy for a node $v \neq d$ is
a weight function $w_v: \pi(v) \rightarrow [0,1]$ that satisfies the
unity and tree conditions, and the preference relation among the
strategies of a node $v$ is defined by the lexicographically at least
relation.  (Thus, a node's best response is a lexicographically
maximal flow.)  

We can now show that fractional stable paths solutions are equivalent
to (pure) Nash equilibria in the fractional BGP game.  We note
that~\cite{Kintali} has also independently shown that a fractional
stable paths solution is a Nash equilibrium of a suitably defined
game.  

\junk{
\begin{definition}A \emph{Lexicographic Ordering} of paths in a graph, given BGP preferences from a node $s$ to a single destination $t$, is the set of paths from $s$ to $t$ that are in the preference list for $s$, ordered from the highest to lowest preference.
\end{definition}

Given capacities for all paths from nodes $u \ne s$ to $t$, consider
the following algorithm to add weights to all $s$-$t$ paths:

\begin{definition}A \emph{Maximally Lexicographic Flow} is the set of weights on all paths from $s$ to $t$ defined by the above algorithm.
\end{definition}
}

\begin{theorem}
\label{thm:equivalence}
A fractional paths solution is stable iff it is
lexicographically maximal for every node. 
\end{theorem}

\begin{proof}
\BfPara{A stable paths solution is a lexicographically maximal flow} 
Let $w$ be a fractional stable paths solution.  Assume, for the sake
of contradiction, that $w_v$ is not lexicographically maximal with
respect to $w_{-v}$.  Then there exists an assignment $w'_v$ such that
$(w'_v, w_{-v})$ is feasible and $w'_v$ is lexicographically greater
than $w_v$ with respect to $w_{-v}$.  Among all such assignments, we
set $w^*_v$ to be an assignment such that the preference of the
highest preference path at which $w_v$ and $w^*_v$ differ is smallest.
Let $P$ be the highest preference path at which they differ and let
${\cal P}$ denote the set of all paths with the same preference as
$P$.

By the definition of stability, at least one of the two stability
conditions must hold for $P$ in $w_v$.  First, assume (S1) is
satisfied.  Then we have $\sum_{P' \in \pi(v)} w_v(P') = 1$ and each
$P'
\in \pi(v)$ with $w_v(P') > 0$ is such that $P' \prefer P$.  This implies
that $\sum_{P' \prefer P} w_v(P') = 1$.  But since $w^*_v$ satisfies the
unity condition, we have $\sum_{P' \prefer P} w^*_v(P') = 1$.  \junk {Since $w_v$
agrees with $w^*_v$ on all paths preferred more than $P$, it follows
that $w^*_v$ is not lexicographically greater than $w_v$, a contradiction.} However, this means that $w_v$ is lexicographically at least $w^*_v$ (by definition of ``lexicographically at least''), so $w^*_v$ is not lexicographically greater than $w_v$, a contradiction.

If (S1) does not hold for $P$, then condition (S2) must be satisfied
for each path in ${\cal P}$.  For each $Q \in {\cal P}$, there exists
a proper suffix $S_Q$ (say with start node $u$) of $Q$ such that
$\sum_{P' \in \pi(v,S_Q)} w_v(P') = w_u(S_Q)$, and each $P' \in
\pi(v,S_Q)$ with $w_v(P') > 0$ is such that $P' \prefer_v P$; for each $Q$,
we set $S_Q$ to be the smallest such suffix.  By our choice of $S_Q$
for each $Q$, we obtain that for $Q, Q' \in {\cal P}$, $\pi(v,S_Q)$
and $\pi(v,S_Q')$ are disjoint if $Q \neq Q'$.  Since $(w^*_v,
w_{-v})$ satisfies the tree condition, we have $\sum_{P' \in
\pi(v,S_Q)} w^*_v(P') \le w_u(S_Q)$.  Therefore, using the fact that $\pi(v,S_Q)$'s are all disjoint, $\sum_{Q \in {\cal P}} \sum_{P' \in \pi(v,S_Q): P' \geq P} w^*_v(P')  \leq \sum_{Q \in {\cal P}} w_u(S_Q) = \sum_{Q \in {\cal P}} \sum_{P' \in \pi(v,S_Q): P' \geq P} w_v(P')$. Furthermore, since $w^*_v$ is
identical to $w_v$ on all paths more preferred than the paths in
${\cal P}$, we obtain that $\sum_{Q \in {\cal P}} w^*_v(Q) \le \sum_{Q
\in {\cal P}} w_v(Q)$.

We now consider two cases.  If $\sum_{Q \in {\cal P}} w^*_v(Q) <
\sum_{Q \in {\cal P}} w_v(Q)$, then $w_v$ is lexicographically greater than
$w^*_v$, leading to a contradiction.  Otherwise, we derive a new
assignment $w'_v$ that is identical to $w^*_v$, except on paths in
${\cal P}$, where it is identical to $w_v$.  This new assignment
$w'_v$ is lexicographically greater than $w_v$, since $w^*_v$ was
lexicographically greater; the highest preference path at which it
differs from $w_v$, however, has lower preference than that for
$w^*_v$, contradicting our choice of $w^*_v$.

\medskip

\BfPara{A lexicographically maximal flow is a stable paths solution} 
Let $w_v$ be a lexicographically maximal flow with respect to
$w_{-v}$.  Consider any path $Q$ that starts at a node $v$.  Suppose,
for the sake of contradiction, $Q$ does not satisfy either of the two
stability conditions.  That is, we have (i) $\sum_{P \prefer Q} w_v(P)
< 1$, and (ii) for each proper suffix $S$ of $Q$ with start node $u$,
we have $\sum_{P \in \pi(v,S), P \prefer Q} w_v(P) < w_u(S)$.  We
derive a new assignment $w'_v$ which is identical to $w_v$ except for
the following: $w'_v(Q) = w_v(Q) + \eps$, for a suitably small $\eps >
0$; for each proper suffix $S$ of $Q$, if there exists a path that is
less preferred than $Q$, shares $S$, and has positive weight, then we
select one such path $P$ and set $w'_v(P) = w_v(P) - \eps$.  It is
easy to see that $w'_v$ satisfies the unity and tree conditions.
However, $w'_v$ is lexicographically greater than $w_v$, a
contradiction.
\end{proof}

\junk{
Let $S$ = the set of path weights in the stable paths solution. Step
through the greedy algorithm described above for finding the maximally
lexicographic flow. Let $P$ be the first path considered by the
algorithm that is assigned a different weight by the algorithm then
the weight in $S$.

If $P$ is assigned a lower weight by the algorithm, then at this point in the algorithm, there must be some suffix of $P$ that has weight=capacity. However, since all paths previously considered in the algorithm were assigned the weights from $S$, there is also a suffix of $P$ in $S$ that has weight=capacity (when $P$ is assigned a lower weight than in $S$). This will violate the tree condition of the fractional stable paths problem. So $P$ could not have been assigned a lower weight by the algorithm.

If $P$ is assigned a higher weight by the algorithm: for path $P$, at least one of the two stability conditions must hold in $S$. First assume condition 2 is statisfied. This means that all paths that start at $s$, share a proper final segment with $P$, and have any weight assigned to them in $S$ must have higher priority than $P$ (so they were considered before $P$ in the algorithm and therefore have identical weights in the algorithm as in $S$). However, it must also be true that for some proper final segment of $P$, the sum of the weights of all of these higher priority paths must add to the weight of the proper final segment. In other words, it must be true that some suffix has weight = capacity by the time we add $P$ in the algorithm, contradicting the fact $P$ had a higher weight in the algorithm.

So $P$ must satisfy stability condition 1 in $S$.  In other words, all paths with any weight assigned in $S$ have higher priority than $P$ (so all paths were already assigned the same weight by the algorithm as the weight assigned in $S$), and the sum of all the path weights from $s$ = 1. However, if all paths were already assigned their weights earlier in the algorithm, and the path weights already sum to 1, the algorithm would not have added $P$. Therefore, a stable paths solution is also a maximally lexicographic flow.
}
\junk{
\BfPara{A maximally lexicographic flow is a stable paths solution.} 
Consider any path $P$ that is given weight between 0 and 1 by the
lexicographic flow algorithm. The amount of weight given to $P$ was
enough to ensure that some subpath had weight=capacity (including the
weight of $P$ and the weights of all paths already assigned weight by
the algorithm). All paths that already have weight have higher
priority than $P$. This meets stability condition 2.

Consider any path $P$ that is given weight 1 by the lexicographic flow
algorithm. Clearly, the sum of the weights of all paths from $s$ to
$t$ is now 1. No other path has any weight assigned, so this meets
stability condition 1.

Consider any path $P$ that is given weight 0 by the lexicographic flow
algorithm. This path is given weight 0 either because the algorithm
completes before it gets to $P$ or because some subpath already has
weight=capacity before $P$ is added. If the algorithm completes before
it gets to $P$, then the sum of the weights of paths from $s$ to $t$ =
1, and all paths with weight have higher priority then $P$, so this
meets stability condition 1. If some subpath already has
weight=capacity, then all of the weight on this subpath must be from
higher priority paths, so this meets stability condition 2.

Therefore, all paths meet a stability condition, and the flow is also
a stable paths solution.  }

\junk{
We are now ready to establish the existence of a Nash equilibrium in
every fractional BGP game, thus yielding an alternative proof to the
one given in \cite{HaxellWilfong} of the existence of a fractional
stable paths solution.
\begin{theorem}
Every instance of the fractional BGP game has a pure Nash equilibrium.
\end{theorem}
\begin{proof}
A game has a pure Nash equilibrium if the strategy space of each
player is a compact, non-empty, convex space, and the best response of
each player $v$ is upper-hemicontinuous on the strategy space of all
players, quasi-concave in the strategy space of $v$, and
non-empty~\cite[Proposition~20.3]{OsborneRubinstein}.

The strategy space of each player $v$ is the set of all weight
functions $w_v$ that satisfy the unity and tree conditions.  This
space is clearly non-empty (it contains the all-zeros strategy),
convex (both the conditions yield linear constraints), and compact.
We next need to show that the best response function of each player is
convex and continuous.  Consider any two best response strategies of
$v$: $w_v$ and $w'_v$.  Since both are lexicographically maximal, for
any $v$, and any path $P \in \pi(v)$, the sum of weights of all paths
with same preference as $P$ is the same for both $w_v$ and $w'_v$.
Any convex combination of these two strategies is thus also
lexicographically maximal; it is also a valid strategy, thus
establishing the convexity of the best response function.  To show
upper-hemicontinuity, we need to show that for any sequence $(w^n_v,
w^n_{-v}) \rightarrow (w_v, w_{-v})$ such that $w^n_v$ is
lexicographically maximal with respect to $w^n_{-v}$, for all $n$,
$w_v$ is lexicographically maximal with respect to $w_{-v}$.  This is
clearly true since the feasibility conditions are linear inequalities
and the stability conditions (to establish lexicographical maximality)
are linear equations.  

Finally, we need to show that for every $w_{-v}$, there exists $w_v$
that is lexicographically maximal with respect to $w_{-v}$.  This
directly follows from the fact that the set of feasible strategies
with respect to $w_{-v}$ is compact.  Indeed, the following greedy
algorithm yields a lexicographically maximal flow.

\begin{algorithm}[htb!]
\caption{
Lexicographically maximal flow algorithm at node $v$
}
\begin{algorithmic}[1]\label{algo:A}
\STATE set $w_v(P)$ to zero for each $P$ in $\pi(v)$
   \FOR{each path $P$ in $\pi(v)$ in preference order $\prefer_v$}
      \IF{there is no proper final segment $S$ (with start node $u$)
      of $P$ such that $\sum_{P \in \pi(v,S)} w_v(P) = w_u(S)$} 
	\STATE
      	increase $w_v(P)$ until $\sum_{P \in \pi(v)} w_v(P) = 1$ or
      there exists a proper final segment $S$ (with start node $u$) of
      $P$ such that $\sum_{P \in \pi(v,S)} w_v(P) = w_u(S)$.  
\ENDIF 
\ENDFOR
\end{algorithmic}
\end{algorithm}
\end{proof}
}

\subsection{Fractional BBC}
\label{sec:definitions.BBC}
\newcommand{\budgetnoargs}{b}
\newcommand{\budget}[1]{b(#1)}
\newcommand{\length}[3]{\ell_{#1}(#2,#3)}
\newcommand{\lengthonearg}[1]{\ell_{#1}}
\newcommand{\cost}[2]{c(#1,#2)}
\newcommand{\costnoargs}{c}
\newcommand{\dist}[1]{\mbox{cost}(#1)}
\newcommand{\distnoargs}{\mbox{cost}}
\newcommand{\pref}[2]{w(#1,#2)}
\newcommand{\prefnoargs}{w}
\newcommand{\dest}{d}
We define a fractional variant of the Bounded Budget Connection game, as
in~\cite{LaoutarisPRST}.  A fractional Bounded Budget
Connection game (henceforth, a fractional BBC game) is specified by a
tuple $\langle V, \dest, \costnoargs,\budgetnoargs\rangle$, and a
length function $\lengthonearg{u}$ for each $u \in V$, where $V$ is a
set of nodes, $d \in V$ is a distinguished destination node,
$\costnoargs: V \times V \rightarrow \mathbb{Z}$, $\budgetnoargs: V
\rightarrow \mathbb{Z}$, and $\lengthonearg{u}: V \times V \rightarrow
\mathbb{Z}$ (for each $u
\in V$) are functions.  For any $u, v \in V$, $\cost{u}{v}$ denotes the 
cost to $u$ of directly linking to $v$, and $\length{x}{u}{v}$ denotes
the length of the link $(u,v)$ from the perspective of $x$, if $u$ has
established this link.  For any node $u \in V$, $\budget{u}$,
specifies the budget $u$ has for establishing outgoing directed links:
the sum of the costs of the links established by $u$ times the amount
placed on each link should not exceed $\budget{u}$.

A strategy for node $u$ is a weight function $w_u: V \rightarrow
[0,1]$ that $u$ places on each outgoing edge $(u,v): v \in V$ such
that $\sum_{(u,v)} \cost{u}{v} \times w_u(v) \le \budget{u}$.  Let
$w_u$ denote a strategy chosen by node $u$ and let $W = \{w_u: u \in
V\}$ denote the collection of strategies.  The network formed by $W$
is simply the directed, capacitated complete graph $G(W)$, in which
the capacity of the directed edge $(u,v)$ is $w_u(v)$.  The utility of
a node $u$ is given by $- f(u)$, where $f(u)$ is the cost of a 1-unit
minimum cost flow from $u$ to $d$, according to the capacities given
by $W$ and the lengths from the perspective of $u$ given by
$\lengthonearg{u}$.  We assume that there is also
always an additional edge from each node to $d$ with cost 0, capacity
$\infty$, and length = some large integer $M \gg n\max_{x,u,v}
\length{x}{u}{v}$; we refer to $M$ as the {\em disconnection
penalty}. In other words, if the max flow from $u$ to $v$ is $\alpha <
1$, then $f(u)$ is the cost of the minimum cost $\alpha$ flow from $u$
to $d$ plus $(1-\alpha) \cdot M$.

\newcommand{\fraccost}[3]{\mbox{cost}_{#2#3}(#1)}
\junk{
\begin{theorem}
\label{thm:fractional}
Every instance of the fractional BBC game has a pure Nash equilibrium.
\end{theorem}

\begin{proof}
A game has a pure Nash equilibrium if the strategy space of each
player is a compact, non-empty, convex space, and the utility function
of each player $u$ is continuous on the strategy space of all players
and quasi-concave in the strategy space of
$u$~\cite[Proposition~20.3]{OsborneRubinstein}.

The strategy space of each player $u$ is simply the convex polytope
given by $\{\sum_v x_v \cost{u}{v} \le \budget{u}\}$.  It is
clearly compact, non-empty, and convex.

The continuity of the utility function is also clear.  It remains to
prove that the utility function of each player is quasi-concave in the
strategy space of the player.  Since the utility function is merely
the negative of the cost, we will show that the cost of the min-cost
flow is a quasi-convex function.

Consider two strategies $X_u$ and $Y_u$ of the player $u$, given fixed
strategies $X_{-u} = \{x_{(s,t)} : s \neq u\}$ of all other players.  Fix a destination $v$.
Suppose there exists a unit flow $f_{(u,v)}$ in $G(X_u)$
and a unit flow $g_{(u,v)}$ from $u$ to $v$ in $G(Y_u)$ with costs $\fraccost{X_u}{u}{v}$ and $\fraccost{Y_u}{u}{v}$ respectively.  Given any
$\lambda \in [0,1]$, consider the strategy $Z_u = \lambda X_u + (1 -
\lambda) Y_u$ for player $u$.  We define the unit flow $h_{(u,v)}$ from
$u$ to $v$ as follows.  For any edge $e = (s,t)$, we set
$h_{(u,v)}(e) = \lambda f_{(u,v)}(e) + (1 - \lambda)g_{(u,v)}(e)$.  Since $f_{(u,v)}(e)$ is at
most $x_{(s,t)}$ and $g_{(u,v)}(e)$ is at most $x_{(s,t)}$, $h_{(u,v)}(e)$ is at most
$\lambda x_{(s,t)} + (1 - \lambda)x_{(s,t)} = x_{(s,t)}$.  Furthermore, the cost
of the flow to node $u$ equals
\begin{eqnarray*}
&& \sum_{(x,y)} h_{(u,v)}(x,y) \length{u}{x}{y} \\
& = & \sum_{(x,y)} \length{u}{x}{y} (\lambda f_{(u,v)}(x,y) + (1 - \lambda) g_{(u,v)}(x,y)) \\ 
& = & \sum_{(x,y)} \length{u}{x}{y} \lambda f_{(u,v)}(x,y) + \sum_{(x,y)} \length{u}{x}{y} (1 - \lambda) g_{(u,v)}(x,y) \\ 
& = & \lambda \fraccost{X_u}{u}{v} + (1 - \lambda) \fraccost{Y_u}{u}{v} \\ 
& \le & \max\{\fraccost{X_u}{u}{v},\fraccost{X_u}{u}{v}\}.
\end{eqnarray*}
Thus, the cost function with respect to one destination is
quasi-convex.  Since the cost for a player is simply the sum of the
weighted costs with respect to all destinations, the quasi-convexity
of the cost function and, hence, the quasiconcavity of the utility
function follow.  This completes the proof of the theorem.
\end{proof}
}

\section{Hardness of Finding Equilibria}
\label{sec:hardness}
In this section, we define a very simple game, the \emph{\gamename}, which is a
special case of both fractional BGP and fractional BBC games. In
Section~\ref{sec:nonconvex}, we show that the set of all equilibria in
a \gamename { }is not convex, implying that we cannot hope to find an
equilibrium for fractional BGP or BBC games using convex
programming. We next present, in Section \ref{sec:ppad-hardness}, our
main result: it is \textbf{PPAD}-hard to find an equilibrium in
the \gamename. Finally, in Section \ref{sec:approximate}, we define an
$\epsilon$-approximate equilibrium for the \gamename, which encompasses two
previously-defined notions of approximation for fractional BGP.  We
extend our \textbf{PPAD}-hardness result to approximate equilibria, thereby
proving that there are no fully polynomial-time approximation schemes
(unless \textbf{PPAD} is in \textbf{FP}) for computing equilibria in
both fractional BGP games and fractional BBC games.
\junk{
In this section, we explore the complexity of finding equilibria in
fractional BGP and BBC games, whose existence is established in
Section~\ref{sec:definitions} and previous work.  A natural first
question to ask is whether the set of all equilibria is convex, as we
might might then hope to find an equilibrium using convex programming.
We present in Section~\ref{sec:nonconvex} a simple fractional BGP
instance for which the set of equilibria is not convex.  We also show,
however, that, unlike Nash equilibria for matrix games, all instances
of fractional BGP have at least one solution in which all weights are
rational.  We next present, in Section~\ref{sec:ppad-hardness}, our
main result that it is \textbf{PPAD}-hard to find an equilibrium in fractional
BGP games.  In Section~\ref{sec:approximate}, we consider two notions
of approximation for equilbria in fractional BGP games, and extend our
\textbf{PPAD}-hardness result to both notions.  Finally, in
Section~\ref{sec:BBC}, we extend these hardness results to fractional
BBC games by showing that the fractional BGP instances used in the
reductions can, in fact, be reduced to fractional BBC games.
}
\subsection{\gamenameCaps}
We begin by defining \gamenamePlural.  In a \gamename\ with a set $S$ of
players, each player's strategy set is $S$.  Each player $i \in S$ has
a preference relation $\prefer_i$ among the strategies. Each player
$i$ chooses a {\em weight distribution}, which is an assignment $w_i: S
\rightarrow [0,1]$ satisfying two conditions: (a) the weights add up
to $1$: $\sum_{j \in S} w_i(j) = 1$; and (b) the weight placed by $i$
on $j$ is no more than the weight placed by $j$ on $j$: $w_i(j) \le
w_j(j)$ for all $i,j \in S$.  As in the case of fractional BGP, the
preference relations $\prefer_i$ induce a preference relation among
the weight distributions as follows: $w_i$ is {\em lexicographically
at least}\/ $w'_i$ if for all $j \in S$, $\sum_{k \prefer_i j} w_i(k)
\ge \sum_{k \prefer_i j} w'_i(k)$.  An equilibrium in a \gamename\ is an
assignment $w = \{w_i: i \in S\}$ such that $w_i$ is lexicographically
maximal for all $i \in S$.  

\newcommand{\pInstance}{$\textbf{P}$}
\newcommand{\bInstance}{$\textbf{B}$}
\newcommand{\ib}{i}
\newcommand{\jb}{j}

We now show that the \gamename { }is a special case of both fractional
BGP and fractional BBC.  

\begin{lemma} \label{BGPinstanceofSimple}
There is a polynomial-time reduction from the \gamename{ }to the fractional BGP game. 
\end{lemma}

\begin{proof}
Consider any instance \pInstance { }of the
\gamename, consisting of a set of players $S$ and a preference
relation $\prefer_i$ for each $i \in S$.  We will create an instance
\bInstance { }of fractional BGP. For each player $i \in S$, create a
node $\ib$ in \bInstance. Also create a universal destination node
$d$. For all $\ib \ne d$, define $P_{ii} = $ the path $(\ib, d)$. For
all $\ib, \jb \neq d$, define $P_{ij} = $ the path $(\ib, \jb,
d)$. For all $\ib$: define $\pi(\ib)$ (the set of $\ib$'s preferred
paths in \bInstance) as the set $\{P_{ij}: \jb \prefer_{\ib}
\ib\}$. If $k \prefer_i j$ in \pInstance, then $P_{ik} \prefer_{\ib}
P_{ij}$ in \bInstance.

Consider any feasible solution $w = \{w_i\}$ for \pInstance, and
define weights $w' = \{w'_{\ib}\}$ for \bInstance { } where $\forall
i, j$, $w'_{\ib}(P_{ij}) = w_i(j)$, and $w'_{\ib}(P) = 0$ for all
other paths $P$. Because $w$ is feasible, for any $i \in S$, $\sum_{j
\in S} w_i(j) = 1$. Therefore, for all $\ib$, $\sum_{\textrm{paths }P}
w'_{\ib}(P) = \sum_{P_{ij}: j \in S} w'_{\ib}(P_{ij}) = \sum_{j \in S}
w_i(j) = 1$, and the unity condition for \bInstance { }is satisfied.
Also, for all $i,j \in S$, $w_{i}(j) \le w_{j}(j)$. Therefore, for all
paths $\ib$ and all paths $P$ starting at $\jb$, $\sum_{P_{ij} \in
\pi(\ib, P)} w'_{\ib}(P_{ij}) = w'_{\ib}(P_{ij}) = w_i(j) \leq w_j(j)
= w'_{\jb}(P_{jj})$, and the tree condition for \bInstance { } is
satisfied, and $w'$ is feasible.

Consider an equilibrium $w' = \{w'_{\ib}\}$ of \bInstance, and define
weights $w = \{w_i\}$ for \pInstance { }where $\forall i, j$, $w_i(j)
= w'_{\ib}(P_{ij})$. Since this is an equilibrium of \bInstance, it
must be feasible and lexicographically maximal. Because it is
feasible, for each node $\ib$, $\sum_{P_{ij} \in \pi(\ib)}
w'_{\ib}(P_{ij}) \leq 1$. This implies that in our new solution for
\pInstance, $\sum_{j: j \prefer_i i} w_j \leq 1$. Because it is
lexicographically maximal, $P_{ii} \prefer_{\ib} P_{ij} \Rightarrow
w'_{\ib}(P_{ij}) = 0$, so $\sum_{j \in S} w_i(j) = \sum_{j: j \prefer_i
i} w_i(j) \leq 1$. Futhermore, since $P_{ii} \prefer_{\ib}$ the empty
path, $w'_{\ib}(P_{ii}) = 1 - \sum_{P_{ij} \in \pi(\ib), j \neq i}
w'_{\ib}(P_{ij})$, and $\sum_{j \in S} w_i(j) = 1$, as required in
the \gamename. $w'$ is feasible also implies for each node $\ib$, and each
path $P$ with start node $j$, we have $\sum_{P_{ij} \in \pi(\ib, P)}
w'_{\ib}(P_{ij}) \le w'_{\jb}(P)$. However, $\{P \in \pi(\jb)\} \cap
\{P: \pi(\ib, P) \neq \emptyset\} = \{P_{jj}\}$, by definition of the
preference sets. So $\sum_{P_{ij} \in \pi(\ib, P)} w'_{\ib}(P_{ij}) =
0$ unless $P=P_{jj}$, and $\{P \in \pi(\ib, P_{jj})\} = \{P_{ij}\}$, so if
$P=P_{jj}$ then $\sum_{P_{ij} \in \pi(\ib, P_{jj})} w'_{\ib}(P_{ij}) =
w'_{\ib}(P_{ij}) \le w'_{\jb}(P_{jj})$. Therefore, $w_{i}(j) \le
w_{j}(j)$, as required for feasibility in \pInstance.

Now, consider any other feasible assignment $\overline{w} = \{\overline{w}_i\}
\cup \{w_j : j \ne i\}$ for \pInstance. Define $\overline{w}'_{\ib}(P_{ij}
= \overline{w}_i(j)$. Then $\overline{w}' = \{\overline{w}'_{\ib}\} \cup \{w'_{\jb} :
j \ne i\}$ is feasible for \bInstance, as shown above, and
lexicographic maximality of $w'$ says that for every path $P_{ij}$ in
$\pi(\ib)$, $\sum_{P_{ik} \prefer_{\ib} P_{ij}} w'_{\ib}(P_{ik}) \ge
\sum_{P_{ik} \prefer_{\ib} P_{ij}} \overline{w'_{\ib}}(P_{ik})$. Therefore,
for every $j \in S$, $\sum_{k \prefer_i j} w_i(k) \ge \sum_{k
\prefer_i j} \overline{w}_i(k)$, so $w$ is also lexicographically maximal,
and $w$ is an equilibrium for \pInstance.

Finally, consider an equilibrium $w = \{w_i\}$ for \pInstance\ and the
weights $w' = \{w'_{\ib}\}$ for \bInstance { } as defined above. From
above, $w'$ is feasible. Consider any other feasible assignment
$\overline{w}' = \{\overline{w}'_{\ib}\} \cup \{w_{\jb} : j \ne i\}$ for
\bInstance. Since $w$ is an equilibrium, it is lexicographically
maximal, so for $\overline{w} = \{\overline{w}_i\} \cup \{w_j : j \ne i\}$
(where $\overline{w}_i(j) = \overline{w}'_{\ib}(P_{ij})$), $\forall j \in S$,
$\sum_{k \prefer_i j} w_i(k) \ge \sum_{k \prefer_i j}
\overline{w}_i(k)$. Therefore, $\forall i,j \in S$, $\sum_{P_{ik}
\prefer_{\ib} P_{ij}} w'_{\ib}(P_{ik}) \ge \sum_{P_{ik} \prefer_{\ib}
P_{ij}} \overline{w}'_{\ib}(P_{ik})$, and $w'$ is also lexicographically
maximal and an equilibrium for \bInstance.
\end{proof}

\begin{lemma} \label{BBCinstanceofSimple}
There is a polynomial-time reduction from the \gamename\ to the fractional BBC game. 
\end{lemma}

\begin{proof}
We use a similar reduction from a \gamename\ to fractional BBC. Given any instance \pInstance\ of the \gamename\, We will create an instance \bInstance { }of fractional BBC = $\langle V,
d, c, b \rangle$, where $V = S$, $d = $ an additional node, $\forall
i,j \in V$: $c(i,j) = 1$, $\forall i$: $b(i) = 1$, plus length
function $l_i$ for each $i \in V$, defined as follows. Let $p_i(k)$ =
the number of $j$ such that $j \prefer_i k$. $\forall j \ne i,
l_i(j,d) = 1$, $l_i(i,j) = p_i(j)$. $\forall j \ne i,k \ne i, l_i(j,k)
= l_i(k,j) = |S| + 1$. $l_i(i,d) = 1 + p_i(i)$. Given a solution to
\bInstance, define a solution to \pInstance: set $w_i(j)$ = the
weight placed on edge $(i,j)$ (for $j \neq i$), and $w_i(i)$ = the
weight placed on edge $(i,d)$.  

Consider any instance \pInstance { }of the \gamename, consisting of a set
of players $S$ and a preference relation $\prefer_i$ for each $i \in
S$. We will create an instance \bInstance { }of fractional BBC =
$\langle V, d, c, b \rangle$, where $V = S$, $d = $ an additional
node, $\forall i,j \in V$: $c(i,j) = 1$, $\forall i$: $b(i) = 1$, plus
length function $l_i$ for each $i \in V$, defined as follows. Let
$p_i(k)$ = the number of $j$ such that $j \prefer_i k$. $\forall j \ne
i, l_i(j,d) = 1$, $l_i(i,j) = p_i(j)$. $\forall j \ne i,k \ne i,
l_i(j,k) = l_i(k,j) = |S| + 1$. $l_i(i,d) = 1 + p_i(i)$. Given a
solution to \bInstance, define a solution to \pInstance\ by setting
$w_i(j)$ = the weight placed on edge $(i,j)$ (for $j \neq i$), and
$w_i(i)$ = the weight placed on edge $(i,d)$.

Since the total cost for all edges is 1, and the total budget for a
node is 1, each node in \bInstance\ will place total weight 1 on edges
adjacent to it. This exactly corresponds to the requirement that
$\sum_j w_i(j) = 1$ in \pInstance. The possible paths for a one-unit
flow from $i$ to $d$ in \bInstance\ are: (1) the path consisting of
only edge $(i, d)$, which has cost $p_i(i) + 1 \leq |S| + 1$, (2) a
path of the form $(i,j,d)$ through some other node $j$, which has cost
$p_i(j) + 1 \leq |S| + 1$, or (3) a path including some edge $(j,k)$
for $j \neq i, k \neq i$, which has cost $ > |S| + 1$. Therefore, a
minimum cost flow will only use paths of the form $(i,d)$ and
$(i,j,d)$, so the requirement in \pInstance\ that $w_i(j) \leq w_j(j)$
corresponds to using the weight $j$ places on edge $(j,d)$ as a
capacity on that edge when finding the min-cost flow. Now, we only
need to show that a node's best response in \bInstance\ exactly
corresponds to a lexicographically maximal weight assignment in
\pInstance.

Suppose we have a best response for node $i$ in \bInstance\ that
corresponds to a weight assignment $w$ in \pInstance\ that is not
lexicographically maximal for $i$. Then, there is some assignment $w'
= w'_i \cup \{w_j: j \ne i\}$ such that for some $j \in S$, $\sum_{k
\prefer_i j} w_i(k) < \sum_{k \prefer_i j} w'_i(k)$. There must be some
$k^+ \in S$ such that $k^+ \prefer_i j$ and $w'_i(k^+) > w_i(k^+)$,
and there must be some $k^- \in S$ such that $\lnot (k^- \prefer_i j)$
and $w'_i(k^-) < w_i(k^-)$. Suppose we move $\epsilon$ weight in the
best response in \bInstance\ from $P_{ik^-}$ to $P_{ik^+}$. $p_i(k^-) >
p_i(k^+)$, so moving this weight will decrease the cost of a minimum
cost flow, contradicting the fact that this was a best response.

Suppose we have a lexicographically maximal weight assignment $w$ for
\pInstance\ that does not correspond to a best response for node $i$ in
\bInstance. Then, in \bInstance, $i$ could move weight from some path
$P_{ij}$ to a different path $P_{ik}$ to decrease the cost of its
min-cost flow. This means that $p_i(k) < p_i(j)$, or the number of
nodes preferred by $i$ over $k$ is smaller than the number of nodes
preferred by $i$ over $j$. Since preference relations are transitive,
this implies that $k \prefer_i j$. However, since $P_{ik}$ had space
left, $w_i(k) < w_k(k)$, so $w$ is not lexicographically maximal.
\end{proof}

\subsection{Non-Convexity}
\label{sec:nonconvex}
\begin{theorem}
There exists an instance of the \gamename { }for which the set of equilibria is not convex.
\end{theorem}
\begin{proof}
\junk{
\begin{figure}[htb]
\begin{center}
\includegraphics[width=3.5in]{nonconvex.eps}
\caption{Example of a fractional BGP game for which the equilibrium
set is not convex.  The preference lists are $a_1$:
  $(a_1a_2t, a_1t)$; $a_2$: $(a_2a_1t, a_2t)$;
  $b_1$: $(b_1b_2t, b_1t)$; $b_2$: $(b_2b_1t, b_2t)$; $c_1$: $(c_1c_2t, c_1t)$; $c_2$:
  $(c_2c_1t, c_2t)$; $x$: $(xa_1t, xb_1t,xc_1t, xt)$.
\label{fig:nonconvex}}
\end{center}
\end{figure}
}
\begin{figure}[htb]
\begin{subfigure} [The $a$ players assign weights $1/2,1/2$, the $b$ players both use $b_1$, the $c$ players both use $c_2$. \label{fig:convex_eq1}] {
\includegraphics[width=2in]{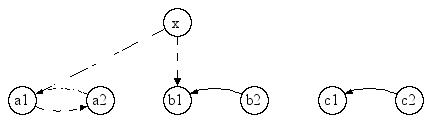}
}
\end{subfigure}
\begin{subfigure} [The $a$ players assign weights $1/2-1/2$, the $b$ players both use $b_2$, the $c$ players both use $c_1$. \label{fig:convex_eq2}] {
\includegraphics[width=2in]{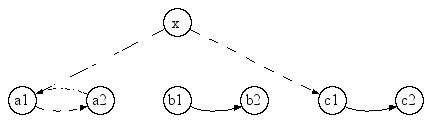}
}
\end{subfigure}
\begin{subfigure} [Combining half of each equilibrium, $x$ will assign $1/2$ to $a_1$, $1/4$ to each of $b_1$ and $c_1$. $x$ could improve by assigning weight only to $a_1$ and $b_1$. \label{fig:convex_comb}] {
\includegraphics[width=2in]{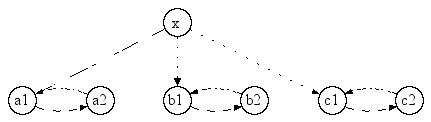}
}
\end{subfigure}
\caption{Example of an instance of the \gamename { }for which the equilibrium
set is not convex.
\label{fig:nonconvex}}
\end{figure}
Consider the following instance of the \gamename.  We have 3 sets of 2
players each, $a_1, a_2$, $b_1, b_2$, $c_1, c_2$, and one additional
player, $x$.  The preference lists for these nodes are: $a_1$: $(a_2,
a_1)$; $a_2$: $(a_1, a_2)$; $b_1$: $(b_2, b_1)$; $b_2$: $(b_1, b_2)$;
$c_1$: $(c_2, c_1)$; $c_2$: $(c_1, c_2)$; $x$: $(a_1, b_1,c_1, x)$.
(Each list gives strategies in order from most preferred to least
preferred.)  We now show two equilibria whose linear combination is
not an equilibrium.  In equilibrium $w$ (figure
\ref{fig:convex_eq1}): $w_{a_1}(a_1) = \frac{1}{2}$, $w_{a_1}(a_2)
=\frac{1}{2}$, $w_{a_2}(a_2)=\frac{1}{2}$, $w_{a_2}(a_1) =
\frac{1}{2}$, $w_{b_1}(b_1) = 1$, $w_{b_2}(b_1) = 1$, $w_{c_1}(c_2) =
1$, $w_{c_2}(c_2) = 1$, $w_x(a_1) = \frac{1}{2}$, $w_x(b_1) =
\frac{1}{2}$.  In equilibrium $w'$ (figure \ref{fig:convex_eq2}):
$w'_{a_1}(a_1) = \frac{1}{2}$, $w'_{a_1}(a_2) =\frac{1}{2}$,
$w'_{a_2}(a_2)=\frac{1}{2}$, $w'_{a_2}(a_1) = \frac{1}{2}$,
$w'_{b_1}(b_2) = 1$, $w'_{b_2}(b_2) = 1$, $w'_{c_1}(c_1) = 1$,
$w'_{c_2}(c_1) = 1$, $w'_x(a_1) = \frac{1}{2}$, $w'_x(c_1) =
\frac{1}{2}$. 
It is easy to verify that $w$ and $w'$ are both equilibria, and in a solution
$\lambda \cdot w + (1 - \lambda) \cdot w'$ (for any $\lambda >
\frac{1}{4}$) (figure \ref{fig:convex_comb} shows $\lambda=\frac{1}{2}$), player $x$ would do better by moving more weight to its second preference. Therefore, the convex combination of $w$ and $w'$ is not an equilibrium.
\junk{
It is easy to verify that $w$ and $w'$ are both equilibria. Consider
$\lambda \cdot w + (1 - \lambda) \cdot w'$ (for any $\lambda >
\frac{1}{4}$) (figure \ref{fig:convex_comb}), and consider the
best response of player $x$. Now, $x$ can still only place weight $\frac{1}{2}$ on its first
preference ($a_1$). However, its second preference, $b_1$,
is now available with weight $\lambda$, but it will assign it only
weight $\frac{1}{4}$. So $x$ would do better to move more weight to
its second preference, thus showing that the convex combination is
not an equilibrium.
}
\end{proof}

\junk{
\begin{theorem}
For any instance of the fractional BGP game and for any instance of
the BBC game, an equilibrium exists in which all players play only
rational weights.
\end{theorem}

We defer the proof of the above theorem to
Section~\ref{sec:personalized}, where we establish the claim for a
larger class of games.
}

\subsection{PPAD Hardness}
\label{sec:ppad-hardness}
We show that finding an equilibrium in \gamenamePlural { }is \textbf{PPAD}-hard.
By our reductions of Lemmas~\ref{BGPinstanceofSimple} and ~\ref{BBCinstanceofSimple}, this immediately implies that finding
equilibria in fractional BGP and fractional BBC games is also
\textbf{PPAD}-hard.  We will follow the framework
of~\cite{DaskalakisGoldbergPapadimitriou}, which shows that finding a
Nash equilibrium in a degree-3 graphical game is \textbf{PPAD}-hard, using a
reduction from the \textbf{PPAD}-complete problem 3-D BROUWER.  In this
problem, we are given a 3-D cube in which each dimension is broken
down into $2^{-n}$ segments -- thereby dividing the cube into $2^{3n}$
cubelets.  We are also given a circuit that takes as input the 3
coordinates of the center of a cubelet (each as an $n$-bit number) and
returns a 2-bit number that represents one of four 3-D vectors: either
$(1,0,0)$, $(0,1,0)$, $(0,0,1)$, or $(-1,-1,-1)$. A solution to the
3-D BROUWER instance is a cubelet vertex such that the set of 8
results obtained by running the circuit on each of the 8 cubelets
surrounding the vertex contains each of the four vectors at least
once.

As in~\cite{DaskalakisGoldbergPapadimitriou}, we will construct a set
of gadgets to simulate various arithmetic operators, logical
operators, arithmetic comparisons and other operators.  We then follow
their framework to systematically combine these gadgets to simulate
the input boolean circuit and to encode the geometric condition of
discrete fixed points in the 3-D BROUWER instance.  In the \gamename {
}we construct, we specify the preference relation of any player $P$ by
an ordered list of a subset of the players, with the last element
being $P$, also referred to as the ``self'' strategy.  When we say
that a player $P$ {\em plays itself}\/ with weight $v$, we mean that
$P$ assigns a weight of $v$ to strategy $P$.  We'll engineer the
payoffs such that the game is only in equilibrium if the weights
assigned by certain players to {\em themselves}\/ successfully echo
the inputs and outputs of 8 copies of the circuit that surround a
solution vertex of the 3-D BROUWER instance.  

For this reduction, we require the following sets of players.

\begin{enumerate}
\item One player for each of the 3 coordinates (the \emph{coordinate
players}). If the graph is an equilibrium, each coordinate player
plays itself with weight equal to its coordinate of the 3-D BROUWER
solution vertex.

\item 
One player for each of the bits of each of the 3 coordinates (the
\emph{bit players}). In order to force these players to correctly
represent the bits, we need some additional players. Assuming we've
correctly calculated the first $i-1$ bits of coordinate $x$ (call them
$x_0, \ldots, x_{i-1}$), we can create the $i^{th}$ bit as
follows. One player will play itself with weight $p_i = x -
\sum_{j=0}^{i-1} \frac{x_j}{2^{j}}$. The bit player will play itself
with weight equal to the $i^{th}$ bit. If $p_i 
\geq \frac{1}{2^i}$, then this bit should be 1. Otherwise, it should
be 0. Therefore, in order to properly extract the bits, we create the
following four types of players.
\newcounter{saveenum}
\begin{enumerate}
\item HALF player: 
In any equilibrium in which a given player plays itself with weight
$a$, the HALF player will play itself with weight $\frac{a}{2}$.
\item DIFF player: 
In any equilibrium in which two given players play themselves
with weights $a$ and $b$, the DIFF player will play itself
with weight $a - b$.
\item VALUE player: 
In any equilibrium, the VALUE player plays itself with weight
$\frac{1}{2}$. This can be easily created by combining a player whose
first preference is itself with a HALF player.
\item LESS player: 
In any equilibrium in which two given players play themselves with
weights $a$ and $b$, respectively, the LESS player plays itself with
weight 1 iff $a \geq b$, and plays itself with weight 0 otherwise.
(Actually, the LESS player we create will be inaccurate if $a$ and $b$ are very close,
which we discuss more below.)
\setcounter{saveenum}{\value{enumii}}
\end{enumerate}
\item 
One player simulates each type of gate used in the circuit of the 3-D
BROUWER instance. For this, we create 3 more types of players.
\begin{enumerate}
\setcounter{enumii}{\value{saveenum}} 
\item AND player: 
In any equilibrium in which two given players play themselves with
weights $a$ and $b$, the AND player will play itself with weight $a
\land b$.
\item OR player: 
In any equilibrium in which two given players play themselves with
weights $a$ and $b$, the OR player will play itself with weight $a
\lor b$.
\item NOT player: 
In any equilibrium in which a given player plays itself with weight
$a$, the NOT player will play itself with weight $\lnot a$.
\setcounter{saveenum}{\value{enumii}}
\end{enumerate}

\item 
Finally, we need to ensure that the graph is in equilibrium if and
only if all four vectors are represented in the results of the 8
circuits.  As in~\cite{DaskalakisGoldbergPapadimitriou}, we will
represent the output of each circuit using 6-bits, one each for $+x,
-x, +y, -y, +z, -z$. Now, the 4 possible result vectors are
represented as $100000$, $001000$, $000010$, and $010101$. We can use
these circuit results with only two additional types of players to
feed back into the original coordinate players.  First, we will create
an OR player for each of the 6 bits (over the 8 vertices), which
yields a result of six 1's if and only if this is a solution vertex.
Therefore, an AND player for each coordinate will all return 1 if and
only if this is a solution vertex; at least one of the coordinates
will be 0 otherwise. We can turn this around using a NOT player for
each coordinate, so that we get all 0's if and only if this is a
solution vertex. Finally, we need the last two new player types, which
we'll use to add these results back to a copy of the original
coordinates (the result will be the original coordinate player).
\begin{enumerate}
\setcounter{enumii}{\value{saveenum}} 
\item 
COPY player: In any equilibrium in which a given player plays itself
with weight $a$, the COPY player will also play itself with weight
$a$.
\item 
SUM player: In any equilibrium in which two given players play
themselves with weights $a$ and $b$, the SUM player will play itself
with weight $\min(a + b, 1)$.
\end{enumerate}
If the coordinates represented a solution vertex to the 3-D BROUWER
instance, then all the values we've added back in will be zero; so the
coordinate players cannot do better by changing their strategies.  On
the other hand, if the coordinates do not form a solution vertex, then
at least one of the values is 1, so that the coordinate player will
have incentive to change strategies and play more weight on itself.
\end{enumerate}

We now describe how to create the
new types of players (gadgets) required for the reduction. For each of
these gadget definitions, we assume we are given a \gamename{ }such
that in any equilibrium, node $X$ plays itself with weight $v_1$ and
node $Y$ plays itself with weight $v_2$.  For the first three gadgets,
we assume $v_1, v_2 \in \{0,1\}$. For the rest of the gadgets, we
assume $v_1, v_2 \in [0,1]$.
\smallskip

\subsubsection*{$\mbox{OR}(X,Y)$}
We can add a new node $R = \mbox{OR}(X,Y)$ that will
play itself with weight $v_1 \lor v_2$ in any equilibrium.  Create a node $R_1$ with
preference list $(X, Y, R_1)$. Let node $R$'s preference list be
$(R_1, R)$.  Now, if $v_1$ and/or $v_2$ is 1, then $R_1$ will play
$R_1$ with weight 0, so $R$ will play itself with weight 1. If both
$v_1$ and $v_2$ is 0, then $R_1$ will play itself with weight 1, so
$R$ will play $R_1$ with weight 1 and $R$ with weight 0.

\subsubsection*{$\mbox{NOT}(X)$}
We can add a new node $N = \mbox{NOT}(X)$ that will play itself with weight $\lnot v_1$ in any equilibrium.  Let node $N$'s preference
list be $(X, N)$.  Clearly, $N$ will play $X$ as much as $v_1$ and
will play $N$ with the remainder.
\subsubsection*{$\mbox{AND}(X,Y)$}
We can add a new node $A
= \mbox{AND}(X,Y)$ that will play
itself with weight $v_1 \land v_2$ in any equilibrium.  Assemble the OR and NOT gadgets
$\mbox{NOT}(\mbox{OR}(\mbox{NOT}(X),\mbox{NOT}(Y)))$.
\subsubsection*{$\mbox{SUM}(X,Y)$}
We can add a new node
$S=\mbox{SUM}(X,Y)$ that will play
itself with weight $max(1, v_1 + v_2)$ in any equilibrium.  Create a node $S_1$ with
preference list $(X, Y, S_1)$. Let node $S$'s preference list be
$(S_1, S)$.  Now, clearly node $S_1$ will play $S_1$ with weight
$max(0, 1-v_1-v_2)$, and node $S$ will play $S_1$ that same
amount. So node $S$ will play itself with weight $1 - max(0,
1-v_1-v_2)$. In other words, if $v_1 + v_2
\geq 1$, then $S$ will play itself with weight 1. Otherwise,
$S$ will play itself with weight $1 - 1 + v_1 + v_2 = v_1 +
v_2$, as desired.
\subsubsection*{$\mbox{DIFF}(X,Y)$}
We can add a new node $D
= \mbox{DIFF}(X,Y)$ that will play
itself with weight $v_1 - v_2$ if $v_1 > v_2$, or 0 otherwise in any equilibrium.  Create
a node $D_1$ with preference list $(X, D_1)$. $D_1$ will play itself
with weight $1-v_1$. Now set the preference list for $D$ to $(D_1,
Y, D)$. $D$ will play itself with weight $\min(0, 1 - (1-v_1) - v_2)
= \min(0, v_1 - v_2)$, as desired.
\subsubsection*{$\mbox{COPY}(X)$}
We can add a new node $C = \mbox{COPY}(X)$ that will play itself with weight $v_1$ in any
equilibrium.  Create a node
$C_1$ with preference list $(X, C_1)$. $C_1$ will play itself with
weight $1-v_1$. Set the preference list for node $C$ to $(C_1,
C)$. $C$ will play $C_1$ with weight $1 - v_1$, leaving weight $v_1$
on $C$.
\subsubsection*{$\mbox{DOUBLE}(X)$}
We can add
a new node $M = \mbox{DOUBLE}(X)$ that
will play itself with weight $\min(1, v_1 * 2)$ in any equilibrium.  Create
player $M_1 = \mbox{COPY}(X)$ and set $M$ as $\mbox{SUM}(X,M_1)$.
\subsubsection*{$\mbox{LESS}(X,Y)$}
Given $\epsilon_l$ $(0 < \epsilon_l
\leq \frac{1}{2}$), We can add a new node $L
= \mbox{LESS}(X,Y)$ to the game that in any equilibrium will play only
itself if $v_1 - v_2 \geq \epsilon_l$, and will play
 $L_1$ (for a new node $L_1$) if $v_1 \leq v_2$.  First
create $D = \mbox{DIFF}(X,Y)$.  Then create $M_1 = \mbox{DOUBLE}(D)$.
For $i=1$ to $-\log \epsilon_l$, create player $M_{i+1} =
\mbox{DOUBLE}(M_i)$.  Call the last DOUBLE player node $L$ and
the extra player for the sum player of the last DOUBLE player node
$L_1$.  If $v_1 \leq v_2$, the DIFF player will return 0, so player
$L$ will play the result of multiplying 0 by 2 many times, or 0. If
$v_1 - v_2 \geq \epsilon_l$, player $L$ will play the max of 1 and
$(v_1 - v_2) * 2^{- \log \epsilon_l} = (v_1 - v_2) *
\frac{1}{\epsilon_l} \geq
\frac{\epsilon_l}{\epsilon_l} = 1$.
\subsubsection*{$\mbox{HALF}(X)$}
We can add
a new node $H = \mbox{HALF}(X)$ that will play itself with weight $v_1/2$ in any equilibrium.  Create a node $H_1$ with
preference list $(X, H_1)$. $H_1$ will play itself with weight $1 -
v_1$. Then create two more nodes: $H_2$ and $H_3$. Node $H_2$ has
preference list $(H_1, H_3, H_2)$. Node $H_3$ has preference list
$(H_1, H, H_3)$. Set the preference list for node $H$ to be $(H_1,
H_2, H)$. Each of $H$, $H_2$, and $H_3$ will use its first choice with
weight $1 - v_1$, leaving $v_1$ for its other two choices.  Then, we
have $w_H(H) + w_H(H_2) = v_1$, $w_{H_2}(H_2) + w_{H_2}(H_3) = v_1$,
and $w_{H_3}(H_3) + w_{H_3}(H) = v_1$. In any equilibrium, it must be
true that $w_H(H_2) = w_{H_2}(H_2)$, $w_{H_2}(H_3) = w_{H_3}(H_3)$,
and $w_{H_3}(H) = w_H(H)$.  Solving this gives $w_H(H) = w_H(H_2) =
w_{H_2}(H_2) = w_{H_2}(H_3) = w_{H_3}(H_3) = w_{H_3}(H) =
\frac{v_1}{2}$.
\medskip

As in~\cite{DaskalakisGoldbergPapadimitriou}, our LESS player plays
the specified action (itself, in our case) with weight 1 if $v_1 \geq
v_2 + \epsilon_l$, and plays itself with weight 0 if $v_1 \leq v_2$,
but will play some unspecified fraction on itself if $v_2 < v_1 < v_2
+ \epsilon_l$.  We use the LESS player to extract the bits
representing the coordinates of a cubelet to be passed into the
circuit.  This procedure is identical to that
of~\cite{DaskalakisGoldbergPapadimitriou}.  Let $X$ denote the
$x$-coordinate player, and let $X_1 = \mbox{COPY}(X)$.  For $i$ from
$1$ through $n$, we create players $B_i = \mbox{LESS}(2^{-i},X_i)$ and
$X_{i+1} = \mbox{DIFF}(X_i, \mbox{HALF}^{i}(B_i))$, where
$\mbox{HALF}^i$ indicates applying the HALF gadget $i$ times. 
\junk{\begin{lemma}
\label{lem:bit_extraction}
Assume $\epsilon_l \ll 2^{-n}$.  Let $x$ be the weight assigned by a
coordinate player to itself.  For $m \le n$, if $\sum_{i = 1}^m b_i
2^{-i} + 2m\epsilon_l < x < \sum_{i = 1}^m b_i 2^{-i} + 2^{-m} -
2m\epsilon_l$ for some $b_1, \ldots, b_m \in {0,1}$, then at any
equilibrium of the \gamename{ }, then $B_m$ assigns $b_m$ to itself and
$X_{m+1}$ assigns $x - \sum_{i = 1}^m b_i 2^{-i}$ to itself.
\end{lemma}
} It can be shown that as long as $x$ is not too close to a multiple
of $2^{-n}$, we will extract its $n$ bits correctly.  If this is not
the case, however, we will not properly extract the bits, and our
circuit simulation may return an arbitrary value.  We resolve this
problem using the same technique as
in~\cite{DaskalakisGoldbergPapadimitriou}: we compute the circuit for
a large constant number of points surrounding the vertex and take the
average of the resulting vectors.  Since these details are almost
identical to that of~\cite[Lemma~4]{DaskalakisGoldbergPapadimitriou},
we omit them.  \junk{
\begin{lemma}
\label{lem:equilibrium}
Let $v_x$, $v_y$, and $v_z$ denote the weights assigned by the three
coordinate players to themselves.  In any equilibrium of the
 \gamename, one of the vertices of the cubelets that
contain $(v_x, v_y, v_z)$ is panchromatic.
\end{lemma}
}
From this reduction, we get:
\begin{theorem}
\label{thm:exact}
It is \textbf{PPAD}-hard to find an equilibrium in a given \gamename. \qed
\end{theorem}
\subsection{Approximate equilibria}
\label{sec:approximate}
Given the hardness of finding exact equilibria in \gamenamePlural{ }(and
fractional BGP and BBC games), a natural next question is whether it
is easier to find approximate equilibria.  We define an
$\epsilon$-equilibrium of a $k$-player \gamename{ }to be a set of
weight distributions $w_1$, \ldots, $w_k$ that satisfy the following
conditions for every player $i$: (a) $\sum_j w_i(j) = 1$; (b) for each
$j$, $w_i(j) \le w_j(j) + \epsilon$; and (c) for each $j$, either
$\sum_{\ell: \ell \prefer j} w_i(\ell) \ge 1 - \epsilon$ or $|w_i(j) -
w_j(j)| \le \epsilon$.  In other words, the weight assigned by a
player $i$ on another player $j$ is at most $\epsilon$ more than the
weight assigned by $j$ on itself; and for any $i$ and $j$, either $i$
plays a total weight of at least $1 - \epsilon$ on players it prefers
at least as much as $j$ or the weight assigned by $i$ on $j$ differs from
that assigned by $j$ to itself by at most $\epsilon$. Note that there exists some threshold preference such that any player preferred strictly more than that must be ``filled'' to within $\epsilon$ of the allowed weight. The rest of at least $1-\epsilon$ weight must be placed on players at the threshold preference. At most $\epsilon$ weight is left for players with preference lower than the threshold.

Two notions of approximation have been defined for fractional BGP: an
$\epsilon$-solution by~\cite{HaxellWilfong} and $\epsilon$-stable solution
by~\cite{Kintali}.  The (polynomial-time) reduction of
Lemma~\ref{BGPinstanceofSimple} mapping a given \gamename{ }instance
\pInstance\ to a fractional BGP game instance \bInstance\ 
has the property that any $\epsilon$-solution or $\epsilon$-stable
solution for \bInstance\ is, in fact, an $\epsilon$-equilibrium for
\pInstance.  This implies that any \textbf{PPAD}-hardness on finding
$\epsilon$-equilibrium for \gamenamePlural{ }immediately yields an equivalent
result for both notions of approximation for fractional BGP.

\junk{We consider two notions of approximation that have been considered for
fractional BGP.  Haxell and Wilfong~\cite{HaxellWilfong} define an
{\em $\epsilon$-solution}\/ to the fractional stable paths problem as
a weight function $w$ that obeys the unity condition, and variants of
the tree and stability conditions, which they call the $\epsilon$-tree and
$\epsilon$-stability conditions.
\begin{itemize}
\item 
$\epsilon$-tree condition: For each node $v$, and each path $S$ with
start node $u$, $\sum_{Q \in \pi(u,P)} w_u(Q) \le \epsilon + w_v(P)$.
\item
$\epsilon$-stability: For any node $v$ and path $Q$ starting at $v$,
either the first stability condition holds or there exists a proper
final segment $S$ of $Q$ with some start node $u$ such that $\sum_{P
\in \pi(v,S)} w_v(P) = w_u(S)$, and for each $P \in \pi(v,S)$ with
$w_v(P) > 0$, we have $P \prefer_v Q$.
\end{itemize}
}
\junk{In other words, a feasible solution is one in which each node chooses
at most 1 unit of flow to $d$ such that no subpath is filled by more
than the amount of flow placed on that subpath by its starting node.
3 constraints: (C1) the sum of $w_v(P)$, over all $P$ In
$\pi(v)$ is at most 1; (C2) for each path $P$ in the graph (starting
at some node $v$), for each node $u$: $\sum_{Q \in \pi(u,P)} w_u(Q)
\le \epsilon + w_v(P)$, and (C3)

for node $u$ and path $Q$ starting at
$u$: either $u$ has placed a total of weight 1 on paths it prefers
more than or the same as $Q$, or there is some subpath $P$ of $Q$ such
that if the initial node of path $P$ has put weight $w$ on $P$, then
$u$ has placed a total of weight $w + \epsilon$ on paths with $P$ as a
subpath that it prefers more than or the same as $Q$.
}

\junk{Kintali ~\cite{Kintali} defines an $\epsilon$-stable solution to the
fractional stable paths problem as one that obeys the unity condition,
the tree condition, and the following condition:
}
\junk{, which we call
each node has a total flow of at most 1 to the destination, (2) for
each path $P$ in the graph (starting at some node $v$), for each node
$u$: $\sum_{Q: P \textrm{ is a subpath of } Q}$ the weight $u$ places
on path $Q$ $\leq$ the weight $v$ places on path $P$, and}
\junk{ for node
$u$ and path $Q$ starting at $u$: either $u$ has placed a total of
weight $w$ ($1-\epsilon \leq w \leq 1$) on paths it prefers more than
or the same as $Q$, or there is some subpath $P$ of $Q$ such that if
the initial node of path $P$ has put weight $x$ on $P$, then $u$ has
placed a total of weight $w$ ($x - \epsilon \leq w \leq x$) on paths
with $P$ as a subpath that it prefers more than or the same as $Q$.

We define an $(\epsilon_h, \epsilon_k)$-equilibrium to fractional BGP
as one that is both an $\epsilon_h$-solution as defined by Haxell and
Wilfong and an $\epsilon_k$-stable solution as defined by Kintali. In
other words, in a $(\epsilon_h, \epsilon_k)$-equilibrium no path is
overfilled by more than $\epsilon_h$, and the solution fills the
lexicographically best paths to within $\epsilon_k$ of the allowed
total weight on each path.
}

\begin{theorem}
\label{thm:bgp_ppad}
It is \textbf{PPAD}-hard to find an $\epsilon$-equilibrium for \gamenamePlural, for
some $\epsilon$ inverse polynomial in $n$.
\end{theorem} 
\begin{proof}
Our proof follows the framework
of~\cite{ChenDengTeng,ChenDengTengJACM} for proving the hardness of
approximating Nash equilibria in 2-player games.  This framework
starts with a high-dimensional discrete fixed point problem, BROUWER,
which is also \textbf{PPAD}-complete.  The input to the problem is a Boolean
circuit that assigns a color from $\{ 1,...,n,n+1\}$ to each interior
node of an $n$-dimensional grid $\{0,1,...,8\}^n$.  This grid has
about $2^{3n}$ cells, each of which is an $n$-dimensional hypercube.
The discrete fixed point is defined to be a panchromatic simplex
inside a hypercube.  This framework
of~\cite{ChenDengTeng,ChenDengTengJACM} uses a new geometric condition
for discrete fixed points, which requires that the average of $n^3$
sampled points in the interior of the targeted panchromatic simplex is
inverse-polynomially close to the zero vector.  The rest of the proof
follows the framework of~\cite{DaskalakisGoldbergPapadimitriou}.

Our broad definition of an $\epsilon$-equilibrium poses additional
technical challenges which did not occur in the reductions of
~\cite{ChenDengTeng,ChenDengTengJACM}.  In particular, in the
presence of errors, our Boolean gadgets only approximately simulate
the Boolean operations, while in previous reductions, the Boolean
gadgets are precise. Therefore, most of our technical effort is to prevent
the magnification of errors in Boolean simulation. In our proof, we
have designed a CORRECTION gadget to accomplish this.
\junk{With the help of this CORRECTION gadget, we can then follow the
techniques of previous
proofs~\cite{ChenDengTeng,ChenDengTengJACM,DaskalakisGoldbergPapadimitriou}.}

We focus on the necessary changes
for the gadgets of Theorem~\ref{thm:exact} to
account for errors, and the description and use of the new CORRECTION
gadget.  Other details closely match those
of~\cite{ChenDengTeng,ChenDengTengJACM,DaskalakisGoldbergPapadimitriou}.

Let $\epsilon_l$ (the measure of the fragility of our LESS gadget) be
a real number such that $\epsilon \le \epsilon_l^3$.  Then, we have
the following error bounds.  

\begin{lemma} \label{lem:logic_bounds}
Assuming node $X$ plays itself with weight $v_1'$, $v_1 - 2\epsilon_l \leq v_1' \leq v_1 + 2\epsilon_l$, and node $Y$ plays itself with weight $v_2'$, $v_2 - 2\epsilon_l \leq v_2' \leq v_2 + 2\epsilon_l$, each of the boolean gadgets defined in the proof of Theorem ~\ref{thm:exact} plays itself within $\pm (4\epsilon_l + 6\epsilon)$ of the correct value for the correct $v_1$ and $v_2$ inputs.
\end{lemma}

\begin{proof}
\subsubsection*{OR}
If $v_1$ and/or $v_2$ is 1, then $v_1'$ and/or $v_2'$ is at least $1 -
2\epsilon_l$, and node $R_1$ will play $R_1$ with weight at most
$2\epsilon_l + \epsilon$, so $R$ will play $R$ with weight at least
$1 - 2\epsilon_l - 2\epsilon$. If both $v_1$ and $v_2$ are 0, then
$v_1'$ and $v_2'$ are at most $2\epsilon_l$, and node $R_1$ will play
$R_1$ with weight at least $1 - 4\epsilon_l - 2\epsilon$, so $R$ will
play $R$ with weight at most $4\epsilon_l + 3\epsilon$.
\subsubsection*{NOT}
If $v_1 = 1$, $v_1'$ is at least $1-2\epsilon_l$, and node $N$ will
play itself with weight at most $2\epsilon_l + \epsilon$. If $v_1 =
0$, $v_1'$ is at most $2\epsilon_l$, and node $N$ will play $N$ with
weight at least $1 - 2\epsilon_l - \epsilon$.
\subsubsection*{AND}
The AND gadget concatenates other new players to
get $\lnot(\lnot v_1 \lor \lnot v_2)$. Each NOT may add at most one
additional $\epsilon$ error to the given value, and the OR may add up
to $3\epsilon$ error (on top of the sum of the errors from both
inputs). So the AND player will return a value within an additive
$4 \epsilon_l + 6\epsilon$ of the correct 0 or 1 answer.
\end{proof}

\begin{lemma} \label{lem:arith_bounds}
Each of the arithmetic gadgets plays itself within $\pm 5\epsilon$ of the correct value for the input it is given.
\end{lemma}

\begin{proof}
\subsubsection*{SUM}
Node $S_1$ will play $S_1$ with weight $w(S_1T) \in [max(0,
1-v_1'-v_2'-2\epsilon), max(0, 1-v_1'-v_2' + 2\epsilon)]$. So node $S$
will play $S$ with weight $w_S(S) \in [v_1' + v_2' - 3\epsilon, v_1' +
v_2' + 3\epsilon]$, unless $w_{S_1}(S_1)=0$, which means $v_1' + v_2'
\geq 1 - 2\epsilon$. In this case, node $S$ will play $S$ with
weight at least $1 - \epsilon$.
\subsubsection*{DIFF}
Node $D_1$ will play $D_1T$ with weight $w_{D_1}(D_1) \in max(0, 1 - v_1' -
\epsilon), max(0, 1 - v_1' + \epsilon)]$. Node $D$ will play $D$ with
weight $w_D(D) \in [max(0, v_1' - v_2' - 3\epsilon), max(0, v_1' - v_2'
+ 3\epsilon)]$, unless $w_{D_1}(D_1)=0$ which means $v_1' \geq 1 -
\epsilon$. In this case, node $D$ will play $D$ with weight at least
$1 - v_2' - 2\epsilon$ and at most $1 - v_2' + \epsilon$ (not
$2\epsilon$ because we cannot underfill the strategy with weight 0).
\subsubsection*{COPY}
Node $C_1$ will play $C_1$ with weight at least $1 - v_1' - \epsilon$
and at most $1 - v_1' + \epsilon$. Node $C$ will play $C$ with weight
at least $v_1' - 2\epsilon$ and at most $v_1' + 2\epsilon$.
\subsubsection*{HALF}
Node $H_1$ will play $H_1$ with weight $w_{H_1}(H_1) \in [1 - v_1' -
\epsilon, 1 - v_1' + \epsilon]$, and each other player will play its
second and third preferences with total weight between $1 -
w_{H_1}(H_1) - \epsilon$ and $1 - w_{H_1}(H_1) + \epsilon$. Each other
player will play itself half of this amount plus or minus $3\epsilon$
(this is easy to verify by writing the system of inequalities and
checking the extreme points). Therefore, node $H$ plays $H$ with
weight at least $\frac{v_1'}{2} - 4\epsilon$ and at most
$\frac{v_1'}{2} + 4\epsilon$.
\subsubsection*{DOUBLE}
The DOUBLE gadget consists of a copy player, which adds at most
$2\epsilon$ error, and a sum player, which adds at most $3\epsilon$
error on top of the sum of the errors in the two inputs. Therefore,
node $M$ plays $M$ with weight at least $2v_1' - 5\epsilon$ and at
most $2v_1' + 5\epsilon$.
\end{proof}

\begin{lemma} \label{lem:lt_bounds}
The LESS player will play itself with weight $< \epsilon_l$ if it is given $v_1', v_2'$ such that $v_1' \leq v_2'$, and with weight $> 1 - \epsilon_l$ if $v_1' - v_2' \geq \epsilon_l$.
\end{lemma}

\begin{proof}
\subsubsection*{LESS}
The LESS gadget inherits its susceptibility to error from its
initial DIFF player (which was, in the exact equilibrium case,
non-zero if and only if $v_1 < v_2$). For the case where $v_1 < v_2$,
we can account for the errors of the DOUBLE players (used to
repeatedly amplify the difference) simply by adding extra iterations
of DOUBLE. Since we stipulated that $\epsilon \leq
\epsilon_l^3$, a value that started $\leq 5\epsilon$ will remain $<
\epsilon_l$, even after doubling enough times to push a value
$\geq \epsilon_l$ to a value over 1 (including extra multiplications
to account for the DOUBLE errors). Therefore, the LESS player will
play itself with weight less than $\epsilon_l$ if $v_1' \leq v_2'$ and
with weight greater than $1 - \epsilon_l$ if $v_1' - v_2' \geq
\epsilon_l$.
\end{proof}
\smallskip

Next, we generate another gadget that can be used to amplify the
results of each boolean logic player before using it, in order to
ensure that each input within the circuit is close to the correct
value.  

\subsubsection*{CORRECTION} 
After a single gate (if the inputs are
within additive $2\epsilon_l$ of the correct 0 or 1 inputs), a player
will play itself at least $1 - 4\epsilon_l - 6\epsilon$ if the correct
answer is $1$, and at most $4\epsilon_l + 6\epsilon$ if the correct
answer is $0$ (based on the analysis in the proof of
Lemma ~\ref{lem:logic_bounds}). Therefore, we need only to add a LESS
player to determine whether or not the result is $< \frac{1}{2}$ and
adjust the value in the correct direction using HALF or DOUBLE
players.  

\begin{lemma}
\label{lem:correction}
By using a CORRECTION gadget after each boolean logic gadget, we can ensure that the output from each gate is at most $2\epsilon_l$ away from the correct output.
\end{lemma}

\begin{proof}
The results of a single gate gadget will be at least $1 -
4\epsilon_l - 6\epsilon$ if the correct answer is $1$, and at most
$4\epsilon_l + 6\epsilon$ if the correct answer is $0$. If the result is $< \frac{1}{2}$, we will add three
HALF players: the first reduces any result (at most $5\epsilon_l$) to
at most $\frac{5\epsilon_l}{2} + 4\epsilon$ (notice that we may add an
additional $4\epsilon$ error from the HALF player), the second reduces
it to at most $\frac{5\epsilon_l}{4} + 6\epsilon$, the third to at
most $\frac{5\epsilon_l}{8} + 7\epsilon$, which is at most
$\epsilon_l$, since $\epsilon \ll \epsilon_l$ . If the result is $>
\frac{1}{2}$, we add a single DOUBLE player, which should give us a
result of at least $1 - \epsilon$ (since the input is very close to 1,
the extra player in the SUM portion of the gadget has to play
0). However, we do collect a small additional error term because of
the LESS used in the CORRECTION player.

We can use the LESS player as an if-statement (as needed above) as
follows: LESS will play one of two strategies with weight close to 1,
the other with weight close to 0. Say $P_1$ is the strategy that will
be played with weight close to 1 ($\geq 1 - \epsilon_l$) if and only
if $v_1 < \frac{1}{2}$, $P_2$ is the strategy played with
weight close to 1 ($\geq 1 - \epsilon_l$) if and only if $v_1 \geq
\frac{1}{2}$. We create the necessary players for both the HALF gadget
and the DOUBLE gadget, but add $P_2$ as the first choice preference
for the three players in the HALF gadget (labeled $H$, $H_2$ and $H_3$
in the gadget description), and add $P_1$ as the first choice
preference for the COPY and SUM players in the DOUBLE gadget (players
$C$ and $S$, but not players $C_1$ and $S_1$). Add one additional
player $\mbox{SUM}(H,D)$, where $H$ is the HALF player and $D$ is the
DOUBLE player (one of the two is playing itself with weight close to
0). To show the correctness of the CORRECTION gadget, consider the
following case analysis, assuming the result we are trying to correct
is value $v \in \{[0,5\epsilon_l), (1-5\epsilon_l, 1]\}$. Call the four players that
make up the DOUBLE gadget $C_1$ (the extra player for the COPY
portion), $C$ (the COPY player), $S_1$ (the extra player for the SUM
portion), and $S$ (the SUM player), and the four players that make up
the HALF gadget $H$, $H_1$, $H_2$ and $H_3$ (as above):
\begin{itemize}
\item[Case 1:] $v \leq 5\epsilon_l$. $C_1$ will play itself
with weight at least $1 - v - \epsilon \geq 1 - 5\epsilon_l -
\epsilon$. $C$ will play $P_1$ with weight at least $1 - \epsilon_l -
\epsilon$. It must play the rest of its weight on the heavily-weighted
$C_1$. $S_1$ will play some amount on the player that has weight $v$
and some on $C$, but must have at least $1 - 2\epsilon_l - 3\epsilon$
left for itself. $S$ will play at least weight $1 - \epsilon_l -
\epsilon$ on $P_1$, and must play the rest of its weight on
heavily-weighted player $S_1$, leaving 0 on itself.

Meanwhile, $H_1$ will play at least $1 - v - \epsilon$ on itself, so
each of $H$, $H_2$ and $H_3$ will use up to within $\epsilon$ of the
weight of $P_2$ (which may be 0), and of the weight of $H_1$ (at least
$1 - v - 2\epsilon$), leaving at most $v + 2\epsilon$ to be divided in
half. As stated above, this remaining amount will be split to within
$\pm 3\epsilon$ across the strategies, so the result will be at most
$\frac{5}{2} \epsilon_l + 4 \epsilon$. Since $\epsilon$ is much
smaller than $\epsilon_l$, the additive $\epsilon$ values with each
iteration of the HALF gadget will be covered by the $\epsilon_l$.

The SUM player in the CORRECTION gadget will return a value at most
the correct sum ($\leq \epsilon_l$ from the previous paragraph) plus
$3\epsilon$.

\item[Case 2:] 
$v \geq 1 - 5\epsilon_l$. $C_1$ will play itself with weight at most
$1 - v + \epsilon \leq 5\epsilon_l + \epsilon$. $C$ will play $P_1$
with weight at most $\epsilon_l + \epsilon$, and will play $C_1$ with
weight at most $5\epsilon_l + 2\epsilon$, leaving at least $1 -
6\epsilon_l + 3\epsilon$ on itself. $S_1$ will try to play at least $1
- 5\epsilon_l - \epsilon$ on the player that has weight $v$ on itself
and at least $1 - 6\epsilon_l - 4\epsilon$ on $C$, which will leave
nothing left for itself. $S$ will play at most $\epsilon_l + \epsilon$
on $P_1$, and at most $\epsilon$ on $S_1$, leaving at least $1 -
\epsilon_l - 2\epsilon$ for itself.

Any errors in the HALF player for this case will be if our player puts
$>0$ weight on the HALF player. However, this will only help to
inflate the final result of the CORRECTION gadget.

The SUM player in the CORRECTION gadget will return a value at least
the correct sum ($\geq 1 - \epsilon_l - 2\epsilon$) minus $3\epsilon$,
or at least $1 - \epsilon_l - 5\epsilon > 1 - 2\epsilon_l$.
\end{itemize}
Using this CORRECTION gadget after each gate, we keep our input values
to within $2\epsilon_l$ of the correct values, as required.
\end{proof}

After the corrections, we're left with the following possible errors
due to the $\epsilon$-approximation. We have small errors in the bit
extraction, which are no larger than the parallel errors
in~\cite{DaskalakisGoldbergPapadimitriou} (they verify that these
small error values will not affect the final result). We also have
small errors (at most $2\epsilon_l$) coming out of the circuit. As
in~\cite{ChenDengTeng,ChenDengTengJACM}, we will repeat the circuit a
polynomial number of times and take the average in order to override
any errors from the LESS gadgets in the bit extraction.  

Taking an average of two results requires 3 steps: first
we divide each ``bit'' in half (we cannot take the average of the
entire values because we have a max value of 1 for any player, so the
average of two 1's would come out to $\frac{1}{2}$). Here, we may pick
up $4\epsilon$ of error for each of the two results. Then, we sum the
two. The total error so far is at most $11\epsilon$. Finally, we take
half of the sum, which also divides the error in half, but may add up
to an additional $4\epsilon$ of error, for a total additional error of
at most $9.5\epsilon$ from taking the average of 2 results.  

We can add CORRECTION gadgets periodically during the averaging and
during the final OR, AND and NOT of the results to keep our total
errors under $2\epsilon_l$. In other words, if this is a solution
vertex for BROUWER, then we will have 6 players, each playing at most
$2\epsilon_l$. If this is not a solution vertex, then at least one of
the 6 players will play at least $1 - 2\epsilon_l$.  \junk{ This does
  not ensure that any fixed point in the BROUWER instance maps to an
  $\epsilon$-equilibrium in the game. However, any
  $\epsilon$-equilibrium in the game does map to a fixed point in the
  BROUWER instance.  }
\junk{
Now, given an $\epsilon$-equilibrium in this game, using the error bounds above, we can trace errors back from the coordinate players to the feedback players and find that each feedback player must be playing a value at most $2\epsilon_l$ on itself. Therefore, the correct value for the feedback player must be 0, so this is a valid fixed point.
}
Suppose we have an $\epsilon$-equilibrium in this game, and the
x-coordinate player is playing value $x$. This is a SUM player, and
the extra player from the SUM gadget must be playing between
$1-x-\epsilon$ and $1-x+\epsilon$. Therefore, the sum of the two
values it is adding (a copy of the coordinate player and the feedback
NOT player) must be between $x-3\epsilon$ (if this player overfills
each of its top stretagies by $\epsilon$) and $x+3\epsilon$ (if this
player underfills each of its top strategies by $\epsilon$). We know
that the copy player must be playing the same value as the coordinate
player to within $2\epsilon$ (between $x-2\epsilon$ and
$x+2\epsilon$). Adding this range to a number $\geq 1 - 2\epsilon_l$
cannot possibly give something in the range $[x-3\epsilon,
x+3\epsilon]$, so the feedback player must be playing a value at most
$2\epsilon_l$ on itself (since we know the feedback player will play
either a value $\leq 2\epsilon_l$ or a value $\geq 1-2\epsilon_l$),
and the correct feedback must be 0, so this is a valid fixed point.

\end{proof}

\junk{
\subsection{Hardness of finding equilibria in fractional BBC games}
\label{sec:BBC}
}Theorem \ref{thm:bgp_ppad} implies that it is \textbf{PPAD}-hard to find an
equilibrium in both fractional BGP and fractional BBC games.  Since it
is \textbf{PPAD}-hard to find a fractional BGP equilibrium, it is natural to
next consider special instances when it might be easier to find an
equilibrium. For instance, in real world internet routing, BGP path
preferences are primarily based on a combination of security
considerations and shortest paths. What would happen if we restrict
ourselves to path preferences that echo the real world? Unfortunately, using only small adjustments to the above hardness proof,
we show that it is \textbf{PPAD}-hard to find an equilibrium even if all preferences are based only on shortest path lengths.
\junk{This analysis can be found in 
\junk{
we next show that it is \textbf{PPAD}-hard to find an equilibrium even if all
path preferences are based only on shortest path lengths.  Proofs of
the following theorems have been deferred to}
Appendix~\ref{app:approximate}. }

\begin{theorem} \label{thm:bgp_shortestPathMetric}
Fractional BGP is \textbf{PPAD}-hard even if each node's preference list
consists of all paths, ordered shortest to longest based on edge
length (where each node defines its own edge lengths, which may not
obey triangle inequality). 
\end{theorem}

\begin{proof}
We will implicitly translate the proof of Theorem \ref{thm:bgp_ppad} to a corresponding proof for BGP, by assuming a destination node $T$, each preference by player $V$ for a player $U$ is now a preference for a path $(V \rightarrow U \rightarrow T)$ (abbreviated $(VUT)$), and each preference by player $V$ for ``self'' is now a preference for path $(V \rightarrow T)$ (abbreviated $(VT)$).
We will add a set of edge lengths for each node in the
gadgets such that the preferences in the gadget
definitions follow shortest path distances according to the specified
lengths.

For each node $U$ used in each of the gadgets, the preference list is
of the form ($UVT$, $UT$), ($UVT$, $UWT$, $UT$), or ($UVT$, $UWT$, $UZT$, $UT$)
(the last is only for the HALF player in the CORRECTION gadget). For
preferences of the first form, we will assign edge lengths $l(UV) = 1,
l(VT) = 1$, and all other lengths are 3. Clearly, to get to $T$
through any node other than $V$, the cost will be greater than 3, so
the direct path will be preferred. The distance via $V$ is 2, so this
will be preferred over the direct path.  For preferences of the second
form, we will assign edge lengths $l(UV) = 1, l(VT) = 1, l(UW) = 2,
l(WT)=1$, and all other lengths are 4. Clearly, the preferences for
the 3 paths in the list will be correctly ordered based on
distance. Any path involving a node other than $V$ or $W$ will have
length greater than 4. Edges $VW$ and $WV$ both also have length 4, so
any path to T that uses $V$ or $W$ (that is preferred over the direct
path) cannot include both $V$ and $W$. This leaves only the paths in
the original preference list. For preferences of the third form, we
will assign edge lengths $l(UV) = 1, l(VT) = 1, l(UW) = 2, l(WT)=1,
l(UZ)=3, l(ZT)=1$, and all other lengths are 5. Similar reasoning
shows that this preserves the preference list.
\end{proof}

\begin{theorem}\label{thm:edge_lengths}
Fractional BGP is \textbf{PPAD}-hard even if all preferred paths are
preference-ordered based on the path length (where each node defines
its own distances on the edge lengths, and these distances form a
metric and obey triangle inequality), assuming we may only use edges
from a given template graph.
\end{theorem} 

\begin{proof}
As in the proof of \ref{thm:bgp_shortestPathMetric}, we will implicitly translate the proof of Theorem \ref{thm:bgp_ppad} to a corresponding proof for BGP, by assuming a destination node $T$, each preference by player $V$ for a player $U$ is now a preference for a path $(VUT)$, and each preference by player $V$ for ``self'' is now a preference for path $(VT)$.

We will add a set of edge lengths for each node in the
gadgets such that the preferences in the gadget
definitions follow shortest path distances according to the specified
metrics.

First, we will replace each direct path with a 2-hop path, by adding
an extra node (whose only preference is for its own direct path). In
other words, we will replace any path of the form $UT$ with a path of
the form $UU'T$. We will replace any use of a direct path, such as
$VUT$, with a use of the modified path: $VUU'T$. We will remove all
other edges straight to $T$ from the template graph, and we remove all
edges into a new node $U'$ except the edge from $U$. Removing edges
straight to $T$ is necessary because a preference list ($VUT$, $VT$)
does not obey triangle inequality for any metric. However, the list
($VUU'T$, $VV'T$) is a valid preference list if $V$ uses the following
edge lengths: $l(VU) = 1, l(UU') = 1, l(U'T)=1, l(VV')=2, l(V'T)=2$
(assuming $VT$ is not allowed). Removing other edges into $U'$ is
necessary because otherwise any path $VU'T$ would have to be preferred
at least as much as $VUU'T$.

Now, for each node $U$ used in each of the gadgets, the preference
list is of the form ($UVT$, $UT$), ($UVT$, $UWT$, $UT$), or ($UVT$, $UWT$, $UZT$,
$UT$). With the new additional nodes, each node now has a preference
list of the form ($UVV'T$, $UU'T$), ($UVV'T$, $UWW'T$, $UU'T$), or ($UVV'T$,
$UWW'T$, $UZZ'T$, $UT$). For the first type of preference list, we will
define the length of each leg of the most preferred path to be 1, the
length of each leg of the second path to be 2, and any other edge in
the graph has length 3. It is easy to verify that these lengths obeys
triangle inequality and give the required preference order. For the
second type of preference list, we will assign edge lengths $l(UV) =
2, l(VV') = 1, l(V'T) = 1, l(UW) = 2, l(WW') = 2, l(W'T)=1, l(UU') =
3, l(U'T) = 3$. In order to ensure triangle inequality, set $l(VW) =
l(WV) = 4$. The rest of the edges in the graph had length 5 (so any
path to the root containing any node other than $U,U',V,V',W and W'$
has length at least 10). In the smaller graph containing only $U, U',
V, V', W, W'$, the paths to the root that haven't been included in the
preferences list or specifically excluded by restricting the edges are
$UVWW'T$ (which has length 9) and $UWVV'T$ (which has length 8) - both
are longer than any path in the preference list. For the third type of
preference list, we will assign edge lengths $l(UV) = 3, l(VV') = 1,
l(V'T) = 1, l(UW) = 3, l(WW') = 2, l(W'T)=1, l(UZ) = 2, l(ZZ') = 3,
l(Z'T) = 2, l(UU') = 4, l(U'T) = 4$. In order to ensure triangle
inequality, set $l(VW) = l(WV) = 6$, $l(VZ) = l(ZV) = l(WZ) = l(ZW) =
5$. The rest of the edges in the graph have length 5.

Since we've added an additional edge to every path, this construction
adds up to $\epsilon$ error for each player in the proof of theorem
\ref{thm:bgp_ppad}. However, these errors will still be overpowered by
the $\epsilon_l$ errors from our LESS gadget, so the proof could
easily be adjusted to compensate.
\end{proof}

Notice, if any edge may be used, and if the preferences are based on
shortest path lengths for a metric defined for each node, then there
is a trivial algorithm for finding an equilibrium: each node only
follows the ``direct to destination'' path. Since a metric must obey
triangle inequality, this path length cannot be strictly longer
(cannot be less preferred) than any path including additional nodes.

\junk{
\begin{theorem}Fractional BGP is in \textbf{PPAD}
\end{theorem}

\begin{proof}
Reduction from fractional BGP to ``end of the line'' (the classical
\textbf{PPAD}-complete problem).

An `` end of the line'' instance consists of a graph with
exponentially (but finitely) many nodes. Each node has in-degree at
most 1 and out-degree at most 1. We are given two polynomial-time
computable functions for each node $v$: \emph{successor(v)}, which
returns the node pointed to by $v$ (or null if $v$ is a sink), and
\emph{predecessor(v)}, which returns the node that points to $v$ (or
null if $v$ is a source). We are also given a source node $s$. A
solution is another node $t \neq s$ such that $t$ is either a source
or a sink.

\end{proof}
}

\section{Existence and Rational Solutions via Personalized Equilibria} ~\label{sec:personalized}
We introduce a new notion of an equilibrium for
matrix games based on min-cost flows.  Because the flow-based payoff
functions enable each player to individually match her distribution to
her opponents' distributions, we call this a {\em personalized
equilibrium}.  We study the structural properties of personalized
equilibria and analyze the complexity of finding such an equilibrium.
We show that both the fractional BGP game and the fractional BBC game
are special cases of matrix games in which players seek a personalized
best response.

\def\R{\mathbb{R}}
\def\Reals#1{\mathbb{R}^{#1}}
\def\payoff#1{\mbox{\rm Payoff (#1)}}
We first define personalized equilibria for two player games.  We then
extend it to multi-player games, including multi-player games with
succinct representations.  Consider a matrix game $(R,C)$ between two
players ROW and COLUMN, in which player ROW has strategies $r_1, r_2,
\ldots, r_{m}$ and player COLUMN has strategies $c_1, c_2, \ldots,
c_{n}$. $R \in \Reals{m\times n}$ is the payoff matrix of ROW, and $C
\in \Reals{m\times n}$ is the payoff matrix of COLUMN.

Like a standard bimatrix game, if player ROW selects $r_{i}$ and
  player COLUMN selects $c_{j}$, the payoff to ROW is $R[i,j]$ and the
  payoff to COLUMN is $C[i,j]$.  Suppose ROW selects a distribution
  $x$ among the strategies $\{r_1, r_2, \ldots, r_{m}\}$, and COLUMN
  selects a distribution $y$ among $\{c_1, c_2, \ldots, c_{n}\}$.
  Unlike payoffs defined for mixed strategies, in which the payoff to
  ROW is $\sum_{i,j} x[i]y[j] R[i,j]$ and the payoff to 
  COLUMN is $\sum_{i,j} x[i]y[j] C[i,j]$, we define the
  payoffs using flows.  The payoff to ROW is:
\vspace{-2mm}
\begin{align}
\payoff{ROW}  = & \quad  \max_{u_{i,j}} \sum_{i,j} u_{i,j} R[i,j] \label{Eqn:ROW}\\
 & \quad \mbox{\bf subject to  } 
  \sum_{j} u_{i,j} = x[i], \quad  \forall i \quad \mbox{and} \quad 
\sum_{i} u_{i,j} = y[j], \quad  \forall j; \nonumber \\
\payoff{COLUMN}  = & \quad  \max_{v_{i,j}}\sum_{i,j} v_{i,j} C[i,j]   \label{Eqn:COLUMN}\\
 & \quad \mbox{\bf subject to  } 
  \sum_{j} v_{i,j} = x[i], \quad  \forall i \quad \mbox{and} \quad 
  \sum_{i} v_{i,j} = y[j], \quad  \forall j. \nonumber
\end{align}
In other words, $\payoff{ROW}$ is the cost of a 1-unit min-cost
  flow from source $r$  to destination $c$ in 
  the directed graph $G_{R}  = (V_{R}, E_{R})$,
  with 
\vspace{-2mm}
\begin{eqnarray*}
 V_{R} & = & \{r,c, r_1, r_2, \ldots, r_{m}, c_1, c_2, \ldots,
  c_{n}\} \\
E_{R} & =  &\{(r \rightarrow r_{i}), \ \forall i\} \cup 
             \{ (r_{i}\rightarrow c_{j}), \ \forall i, j\} \cup 
              \{(c_{j} \rightarrow c),\ \forall j \},
\end{eqnarray*}
where the capacity of edge $(r \rightarrow r_{i})$ is $x[i]$, the
capacity of edge $(c_{j} \rightarrow c)$ is $y[j]$, and the capacity
of all other edges is $+\infty$.  The cost of edge $(r_{i}\rightarrow
c_{j})$ is $-R[i,j]$, and the cost of all other edges is 0.  We note
that for any distributions $x$ and $y$, a unit-flow from $r$ to $c$
always exists, so the above payoff function is well-defined.

Similarily, $\payoff{COLUMN}$ is the cost of a 1-unit minimum-cost
  flow from source $c$  to destination $r$ in the directed graph $G_{C}  = (V_{C}, E_{C})$,
  with
\begin{eqnarray*}
 V_{C} & = & \{r,c, r_1, r_2, \ldots, r_{m}, c_1, c_2, \ldots,
  c_{n}\} \\
E_{C} & =  &\{(c \rightarrow c_{j}), \ \forall j\} \cup 
             \{ (c_{j}\rightarrow r_{i}), \ \forall i, j\} \cup 
              \{(r_{i} \rightarrow r),\ \forall i \},
\end{eqnarray*}
where the capacity of edge $(c \rightarrow c_{j})$ is $y[j]$,
the capacity of edge $(r_{i} \rightarrow r)$ is $x[i]$, and 
the capacity of all other edges is $+\infty$. The cost of edge $(c_{j}\rightarrow r_{i})$ is $-C[i,j]$, and
  the cost of all other edges is 0.


Because there is no condition such as $u[i,j] = v[i,j]$ in
   Eqn. (\ref{Eqn:ROW}), (or in the payoff function for COLUMN) 
  each player can individually choose the best way to match 
  the distributions.
We therefore refer to these payoff functions as \emph{personalized payoff
  functions}, and we call an equilibrium for the game with
   these payoffs a \emph{personalized equilibrium}.
Using personalized payoffs, 
  each player plays a distribution across her strategy space and
\emph{chooses} how to combine it with the strategy distributions of the other players.

In addition to the fractional BGP and BBC games, this concept of personalized equilbria is inspired by the correlated equilibrium of Aumann (\cite{Aumann}).
Recall that the correlated payoff function requires $u_{i,j} = v_{i,j}$
  in  Eqn. (\ref{Eqn:ROW}) and Eqn. (\ref{Eqn:COLUMN}),
  but relaxes Nash's condition of $u_{i,j} = v_{i,j} = x[i]y[j]$.
We are considering payoff functions (personalized payoff functions) which futher relax this by removing $u_{i,j} = v_{i,j}$.

\junk{
The personalized payoff function may help us to capture some
  real-world applications.
For example, imagine a good made up of two type of
  widgets, to be assembled and sold in two locations. 
The manager at each location is in charge of one type of widget. 
She  must decide on
  the details of the widget - perhaps the color, weight, shape, etc.,
  but she has no control over the decisions about the other widget. 
She sends copies of her  widget to the other location and receives copies of the other
  widget from the other location manager. 
Then she assembles the goods in whatever manner will best sell at her
location. 
For instance, the customers in one location may 
 prefer multi-colored goods, so the manager in that
  location mixes red copies of the first widget with green copies of
  the second widget (and vice versa), while the customers in the other
  location might prefer solid colors, so the manager in that
  location mixes red with red and green with green. 
In a sense, this creates an ``everyone wins'' scenerio. 

}

One can extend the personalized payoff functions to multi-player
matrix games.  Suppose we are given a $k$-player matrix game $G$, with
$S_i$ being a set of $m_i$ strategies for player $i$, $1 \le i \le
k$, and $u_i: \prod_j S_j \rightarrow {\mathbb R}$ being the utility
function for player $i$.  As in a mixed strategy, each player $i$
chooses a probability distribution $p_i: S_i \rightarrow [0,1]$ over
the strategies in $S_i$.  Given $p_1$, \ldots, $p_k$, the personalized
payoff function for player $i$ is computed as follows.  Construct a
hypergraph $H_i$ with $V = \cup_j S_j$ as the set of nodes and $E =
\prod_j S_j$ as the set of hyperedges.  Consider a (fractional
hypergraph) matching defined by an assignment $w_i : E \rightarrow
{\mathbb R}$ of weights to each hyperedge that satisfies the condition
that the sum of weights of all hyperedges adjacent to any strategy $s
\in S_j$ (for any $j$) equals $p_j(s)$.  Define the weight of
matching $w_i$ as $\sum_{e \in E} w_i(e) u_i(e)$.  The payoff to
player $i$ is then simply the cost of the maximum-weight matching in
$H_i$.

\junk{Suppose there are $k$ players.  Let $S_i$ denote the strategy set of
the $i$th player, $1 \le i \le k$.  We have one node for each possible
strategy for each player, with a $k$-way hyperedge connecting each set
of one strategy per player.  Each hyperedge (each set of $k$ strategy
nodes, one per player) has an associated payoff for each player. We
can also view the negative of the payoff as a cost to each player for
each hyperedge she uses.

The probability
with which a player chooses a strategy can be viewed as a capacity of
the strategy node.  Each player's objective is to minimize the cost of
the 1-unit minimum-cost hypergraph matching using its own costs
(payoffs) for the hyperedges.  A hypergraph matching is an assignment
of weight to each hyperedge such that the total weight of all edges
adjacent to a node is at most the capacity of the node.  In a
personalized equilibrium, each player's chosen distribution, or
capacities on her own strategies, give an optimal matching for that
player, given the capacities chosen by the other players.
}

The concept of personalized equilibria is extendible to games with succinct
representations such as graphical games \cite{KearnsLittmanSingh}
and multimatrix games \cite{Yanovskaya}. It can also be viewed as a relaxation of correlated
equilibrium, as mentioned above.

\junk{
., 
In a graphical game, each player's
payoff depends only on her own distribution and the distributions of a
small number of other players (her ``neighbors'' in a graph). We can
view a hyperedge used by a player as only connecting a strategy of
that player to one strategy of each of her neighbors, or we can still
consider a hyperedge as connecting one strategy for each player, but
her payoff is identical, regardless of which strategies are used from
the non-neighboring players. We can also consider personalized
equilibria in pairwise games, as in \cite{???}. In a pairwise game,
each player's payoff is a (possibly-weighted) average of the payoffs
as if she is playing a separate two-player game against each other
player. The input to a pairwise game is $\binom{n}{2}$ pairs of
matrices: each gives the payoffs for one two-player game.  In fact,
for matrix games with more than 2 players, there are several possible
personalized payoff functions.  We discuss them in the full version of
this paper.
}

\newcommand{\mInstance}{$\textbf{M}$}
\begin{theorem}
\label{thm:BGPPersonalized}
Finding an equilibrium in the fractional BGP game can be reduced to
finding a personalized equilibrium in a matrix game.
\end{theorem}

\begin{proof}
Consider any instance \bInstance { }of fractional BGP. We will create
a matrix game, \mInstance, such that a solution to the \mInstance {
}is a personalized equilibrium if and only if a corresponding solution
to \bInstance { }is an equilibrium.

For each node $v$ in \bInstance, create a player $v'$
in \mInstance. Assign $v'$ one strategy $P'$ for each path
$P \in \pi(v)$, plus one strategy for ``no path.'' Let $q_v(P)$ = the
number of paths $Q$ such that $P \prefer_v Q$. Next, we will define
the payoff to $v'$ for a hyperedge in \mInstance { }containing $P'$
(for $P \in \pi(v)$). If this hyperedge contains all proper suffixes
of $P$, then the payoff to $v'$ will be $q_v(p) + 1$. Otherwise, the
payoff to $v'$ will be $0$. All hyperedges including the ``no path''
strategy for $v'$ will have payoff $0$ for $v'$.

Given a set of distributions $\{p_{v'}\}$ and a set of hyperedge
weights $w$ in \mInstance, we can assign path weights $w'$
in \bInstance: $w'_v(P) = p_{v'}(P')$. If ``no path'' has any weight,
this weight is not assigned in \bInstance.

We will show that a solution in \bInstance { }is feasible if and only if the corresponding solution to \mInstance { }is feasible, then show the correspondence of equilibria. The unity condition is clearly preserved: the distribution for a node in \mInstance { }is a distribution of $1$ unit. The weights placed on edges in \bInstance { }also sum to $1$. Now let's consider the tree condition. Suppose the tree condition is violated in \bInstance. Then, there exists a path $S$ starting at some node $u$ such that for another node $v$, $\sum_{P \in \pi{v, S}} w'_v(P) > w'_u(S)$. This means that in \mInstance, we had a strategy $S'$ of node $u'$ such that for node $v'$, $\sum_{e \in E: S' \in e} w_{v'}(e) > p_{u'}(S)$, which means the solution to \mInstance { }was also infeasible. Now, suppose we have a solution to \mInstance { }that is infeasible. Then, there is some $S' \in S_{u'}$ such that for some node $v'$, $\sum_{e \in E: S' \in e} w_{v'}(e) > p_{u'}(S')$. If the weight placed on $S'$ was from a path that did not include $s$ as a subpath, $v'$ could move the excess weight from $S'$ onto any strategy of $u'$ without changing the payoff, so all remaining weight on $S'$ much be from paths that contains $S'$ as a suffix. But then, for \bInstance, we have $\sum_{P \in \pi{v, S}} w'_v(P) > w'_u(S)$ - another infeasible solution.

As a first step in showing that the equilibria are equivalent, we will note that fractional BGP preference lists across paths can be replaced with preference weights. Any weights that preserve the $\prefer$ relationship will also preserve the set of equilibria. To show this, first define preference weights $u_v(P)$ for all paths $P \in \pi(v)$ such that $u_v(P) \geq u_v(Q)$ if and only if $P \prefer_v Q$. Now, suppose we have an equilibrium $w$ using weights $u$ which is not lexicographically maximal, plus a lexicographically maximal solution $w'$. Let $\cal{P}$ = the set of paths such that for all $P$ and $Q$ in $\cal{P}$, $P \prefer_v Q$, $Q \prefer_v P$, and for all $Q >_v P$, $w_v(Q) = w'_v(Q)$. Then, we know that 
\begin{enumerate}
\item $\sum_{Q: Q \prefer_v P, P \in \cal{P}} w_v(Q) = \sum_{Q: Q \prefer_v P, P \in \cal{P}}$
\item $u_v(P) = u_v(Q)$ if $P \in \cal{P}$, $Q \in \cal{P}$
\item $u_v(P) > u_v(Q)$ if $P \in \cal{P}$, $Q \notin \cal{P}$. 
\end{enumerate}
Replacing $w$ with $w'$ will keep the same weights on all paths with strictly higher preference weight that $u_v(P), P \in \cal{P}$, increase the weight of paths in $\cal{P}$ by some total increase amount $I$, and decrease the total weight by $I$ of the paths with preference weight $< u_v(P), P \in \cal{P}$, thereby improving the solution. Similarly, if we have an equilibria using preference weights, it must be lexicographically maximal.

Now, we can write a linear program to find a best response for fractional BGP, In this LP, $w_v$ is the set of weights assigned by node $v$.

\begin{eqnarray*}
\max \sum_{P \in \pi(v)} w_v(P) u_v(P) & \\
\sum_{P \in \pi(v,S)} w_v(P) \leq w_u(S) & S \in \pi(u), 1\le u \le k\\
w_v(P) \ge 0 & P \in \pi(v)
\end{eqnarray*}

We will also adjust the linear program for finding a best response using personalized payoffs to work for graphical games. In the graphical representation, we write $e \in E$ to represent a hyperedge, where a hyperedge is a subset of at most one strategy per player (compared to exactly one strategy per player previously). We still use $w_v(e)$ to mean the weight placed by player $v$ on hyperedge $e$, $u_v(e)$ to mean the payoff to player $v$ for hyperedge $e$, and $p_v(s)$ to mean the weight placed by player $v$ on his own strategy $s$. The following linear program defines a best response for player $v$.

\begin{eqnarray*}
\max \sum_{e\in E} w_v(e)u_v(e) & \\
\sum_{e: s \in e} w_v(e) \leq p_u(s) & s \in S_u, 1\le u \le k\\
w_vi(e) \ge 0 & e \in E
\end{eqnarray*}

Now, if we assign preference weights for \bInstance: $u_v(P) = u_v(e)$ where $e$ is the hyperedge in \mInstance { }corresponding to $P$ and all suffixes of $P$, the two linear programs are exactly equivalent. Therefore, the set of equilibria is exactly equivalent.
\end{proof}

\begin{theorem}
\label{thm:BBCPersonalized}
Finding an equilibrium in the fractional BBC game can be reduced to finding a personalized equilibrium in a matrix game.
\end{theorem}

\begin{proof}
Consider any instance of fractional BBC. Create a player in the matrix
game for each node in the BBC instance. Assign the player one action
for each available edge in the BBC instance. For any hyperedge in the
matrix game, a player's payoff is negative of the length of the
shortest path to the destination made up of a subset of the edges
represented by that hyperedge (or negative of the disconnection
penalty if there is no such path to the destination). The proof that
this preserves the set of equilibria is similar to the above proof for
fractional BGP.
\end{proof}

\subsection{Existence and Rational Solutions}
\label{sec:personalized.rational} 
 \junk{Throughout the
  following subsection, we let $G$ denote a $k$-player matrix game
  with $S_i$ being the set of $m_i$ strategies for player $i$, and
  $u_i: \prod_j S_j \rightarrow {\mathbb R}$ being the utility
  function for player $i$, $1 \le i \le k$.  Let $E$ denote $\prod_j
  S_j$; we refer to each element of $E$ as a hyperedge.}
\begin{theorem}
\label{thm:exist}
For every multi-player matrix game, a personalized equilibrium
always exists.
\end{theorem}

\begin{proof}
Given the matrix game $G$, we construct the $k$-player game ${\cal G}$
in which the $i$th player's strategy space is the set of all
probability distribution functions over $S_i$ and the payoff is given
by the personalized payoff function defined above.  Then a
personalized equilibrium of $G$ is equivalent to a Nash equilibrium of
${\cal G}$.  By~\cite[Proposition~20.3]{OsborneRubinstein}, a game has
a pure Nash equilibrium if the strategy space of each player is a
compact, non-empty, convex space, and the payoff function of each
player is continuous on the strategy space of all players and
quasi-concave in the strategy space of the player.  The set of
probability distributions over $S_i$ is clearly nonempty, convex, and
compact.  Furthermore, given probability distributions $p_i$ over
$S_i$, $1 \le i \le k$, the payoff for any player $i$ is simply the
solution to the following linear program with variables $w_i(e)$, over
$e \in E$.
\begin{eqnarray*}
\max \sum_{e\in E} w_i(e)u_i(e) & \\
\sum_{e: s \in e} w_i(e) = p_j(s) & s \in S_j, 1\le j \le k\\
w_i(e) \ge 0 & e \in E
\end{eqnarray*}
It is easy to see that the payoff function is both continuous in the
probability distributions of all players, and quasi-concave in the
strategy space of player $i$, thus completing the proof of the
theorem.
\end{proof}

\begin{theorem}
\label{thm:rational}
For any matrix game with all rational payoffs, there exists a
personalized equilibrium in which the probability assigned by each
player to each strategy is a rational number.
\end{theorem}
\begin{proof}
Let $G$ be a $k$-player matrix game.  (Please refer to the beginning
of Section~\ref{sec:personalized} for relevant notation.) For each
player $i$, let $p_i: S_i \rightarrow [0,1]$ denote a probability
distribution over its strategies.  If $p = (p_1, \ldots, p_k)$ forms a
personalized equilibrium, then it provides a feasible solution to the
following linear program over variables $w_i(e)$, where $e \in \prod_j
S_j$ and $1 \le i \le k$, and $p_i(s)$, where $1 \le i \le k$ and $s
\in S_i$:
\begin{eqnarray}
\sum_{e: s \in e} w_i(e) & = p_j(s) & s \in S_j, 1\le j \le k, 1 \le i
\le k \nonumber\\
\sum_{s \in S_i} p_i(s) & = 1 & 1 \le i \le k \label{eqn:lp1}\\
w_i(e) & \ge 0 & 1 \le i \le k, e \in E \nonumber
\end{eqnarray}
Furthermore, if $p$ is a personalized equilibrium, then each $(p_i,
w_i)$ pair maximizes $\sum_{e} w_i(e) u_i(e)$ subject to
LP~(\ref{eqn:lp1}).  Suppose $p$ is not a personalized equilibrium, yet
satisfies LP~(\ref{eqn:lp1}).  This is so if and only if there exists a
player $\ell$ for which $(p_\ell, w_\ell)$ does not maximize $\sum_{e}
w_\ell(e) u_\ell(e)$.  Suppose $(p'_\ell, w'_\ell)$ with $w'_\ell \neq
w_\ell$ is an optimal choice for player $\ell$.  Then, $\delta =
w'_\ell - w_\ell$ is a feasible solution to the following LP:
\begin{eqnarray}
\sum_{e \in E} \delta(e)u_\ell(e) & > 0 & \nonumber\\
\sum_{e: s \in e} \delta(e) & = 0 & s \in S_j, 1 \le j \le k \label{eqn:lp2}\\
\delta(e) & \ge -w_\ell(e) & e \in E\nonumber
\end{eqnarray}
If $F$ is the set of hyperedges for which $\delta(e)$ is
negative, then $\delta(e)$ satisfies LP~(\ref{eqn:lp2}) only if
$w_\ell(e) > 0$ for all those hyperedges.  This motivates
replacing the last constraint of (\ref{eqn:lp2}) with these
two new constraints:
\begin{eqnarray}
\sum_{e \in E} \delta(e)u_\ell(e) > 0 \nonumber\\
\sum_{e: s \in e} \delta(e) = 0 & & s \in S_j, 1 \le j \le k \label{eqn:lp3} 
\end{eqnarray}
This LP, which we refer to as LP~(\ref{eqn:lp3}) is independent of $w_\ell$, for
each player $\ell$ and $F \subseteq E$.

We have thus argued that $p$ is a personalized equilibrium if and only
if there exists $w = (w_1, \ldots, w_k)$ such that $p$ and $w$ satisfy
LP~(\ref{eqn:lp1}) and, if LP~(\ref{eqn:lp3}) is feasible for some $\ell$ and $F$,
then $w_\ell(e)$ should not be positive for all $e$ in $F$.  We thus
add the following constraints to LP~(\ref{eqn:lp1}):
\[
\min_{e \in F} w_\ell(e) = 0, \mbox{ for all $\ell$ and $F$ such that LP~(\ref{eqn:lp3}) is feasible}.
\]
By taking all combinations of one hyperedge from each of the above
product constraints, we get an exponential number of linear programs
(with all rational coefficients), the union of which precisely
describes all personalized equilibria.  By Theorem~\ref{thm:exist}, at
least one of these linear programs is feasible, which implies that
there exists a personalized equilibrium with all rational
probabilities.
\end{proof}

\subsection{Complexity of finding Personalized Equilibria}

\subsubsection{Two Player Personalized Equilibria}

It is not hard to show that the set of all two-player personalized
  equilibria is convex.
In fact, we can give a stronger characterization, which will lead to
  a polynomial time algorithm.

\begin{theorem}
A 2-player personalized equilibrium can always be found in polynomial time.
\end{theorem}

\begin{proof}
Recall the secondary definition of the personalized payoff to player ROW in a two-player game given at the start of Section \ref{sec:personalized}: 

$\payoff{ROW}$ is the cost of a 1-unit minimun-cost
  flow from source $r$  to destination $c$ in 
  the directed graph $G_{R}  = (V_{R}, E_{R})$,
  with
\begin{eqnarray*}
 V_{R} & = & \{r,c, r_1, r_2, \ldots, r_{m}, c_1, c_2, \ldots,
  c_{n}\} \\
E_{R} & =  &\{(r \rightarrow r_{i}), \ \forall i\} \cup 
             \{ (r_{i}\rightarrow c_{j}), \ \forall i, j\} \cup 
              \{(c_{j} \rightarrow c),\ \forall j \},
\end{eqnarray*}
where the capacity of edge $(r \rightarrow r_{i})$ is $x[i]$, the
capacity of edge $(c_{j} \rightarrow c)$ is $y[j]$, and the capacity
of all other edges is $+\infty$.  The cost of edge $(r_{i}\rightarrow
c_{j})$ is $-R[i,j]$, and the cost of all other edges is 0. 

A similar definition of a flow on a graph $G_{C}$ gives the payoff function for player COLUMN. 

Now, let graph $G = $ the union of $G_{R}$ and $G_{C}$. We will now consider a subgraph $G'= (V',E') \subset G$, such that $V' = V_{R} \cap V_{C}$, $(r_i \rightarrow c_j) \in E_R$ is in $E'$ if and only if $R[i,j] \geq R[i',j]$ for all $i'$, and $(c_j \rightarrow r_i) \in E_C$ is in $E'$ if and only if $C[i,j] \geq C[i,j']$ for all $j'$.

\medskip
\BfPara{Any directed cycle in $G'$ corresponds to a personalized equilibria}
Consider any cycle \\ $\{r_{i1}, c_{j1}, r_{i2}, c_{j2}, \ldots, r_{il}, c_{il}\}$ in $G'$, each node played with weight $\frac{1}{l}$. Player ROW can match each of his strategies $r_{ik}$ with player COLUMN's strategy $c_{jk}$. Since this is a best response for player ROW, ROW cannot
do better by changing to another strategy. Similarly, player ROW can match each of his strategies $c_{jk}$ with player ROW's strategy $r_{i(k+1)}$ for $k < l$, $c_{jl}$ can be matched with $r_{i1}$. 

\medskip
\BfPara{Every personalized equilibria is a linear combination of cycles in $G'$}
Starting with any bipartite graph from $G'$ in which the in-degree
equals the out-degree of each node (a characteristic of any
personalized equilibria), we can remove any cycle (which is a personalized equilibria) and we are still
left with a bipartite graph with the same characteristic. 
\end{proof}

\subsubsection{Multi-player personalized equilibria}
\begin{theorem} \label{thm:multi_pers_not_convex}
For multiplayer games, the set of all personalized equilibria may not be convex. 
\end{theorem}
\junk{
\begin{proof}
We have shown that the \gamename\ is a special case of fractional BGP, and we have shown that fractional BPG is a special case of finding personalized equilibria. Therefore, the non-convex example from Section \ref{sec:nonconvex} also proves this theorem.
\end{proof}
}

\begin{proof}
Consider the following example, with 3 players, 2 strategies per
player. Player $i$ has strategies $a_i$ and $b_i$. Let
$P_i(a_1,a_2,a_3)$ = the payoff to player 1 for hyperedge
$\{a_1,a_2,a_3\}$. The payoffs to player 1 are: $P_1(a_1,a_2,a_3) =
P_1(a_1,b_2,b_3) = 1$, $P_1(b_1,a_2,b_3) = P_1(b_2,b_2,a_3) = 2$, the
other 4 payoffs for player 1 are all 0.  The payoffs for the other
players are 1 for all hyperedges. In this example, the pure strategies
$a_1,a_2,a_3$ and $a_1,b_2,b_3$ are both equilibria. However, a
combination of these two, $a_1=1, a_2=a_3=\lambda,
b_2=b_3=(1-\lambda)$, is not an equilibrium, since player p would
prefer to play hyperedges $\{b_1,a_2,b_3\}$ and $\{b_1,b_2,a_3\}$.
\end{proof}

We have shown that the \gamename\ is a special case of fractional BGP, and we have shown that fractional BPG is a special case of finding personalized equilibria. Therefore, the non-convex example from Section \ref{sec:nonconvex} also proves theorem \ref{thm:multi_pers_not_convex}.
\junk {
For some special formulation of multiplayer personalized games,
  polynomial-time algorithm exists.
As in the pairwise personalized equilibria, each player is playing a
2-player game against each other player. However, in this version, we
are given split values $\lambda_{ij}$ for each pair of players $i$ and
$j$, such that for any player $i$, $\sum_{j \neq i} \lambda_{ij} =
1$. Each player $i$ may divide up his strategy selections into
segments of size $\lambda_{i1}, \lambda{i2}, \ldots, \lambda{in}$ each
of which will be used against the appropriate other player. The
segment to be used against player $j$ will be multiplied by
$1/\lambda_{ij}$ before being used against player $j$. The cost to $i$
is $\sum_{j \neq i} $[$\lambda_{ij} * $ cost of a two-player game
against $j$ using the section of $i$'s strategies that were chosen to
be used against $j$].
}


\junk{
\begin{theorem}A split pairwise personalized equilibria can always be found in polynomial time.
\end{theorem}

\begin{proof} A split pairwise personalized equilibria can be found using the following linear program. \\

Let $x_{ij\alpha\beta}$ = the amount that player $i$ uses edge $(\alpha, \beta)$ against player $j$. \\ 

$\forall$ players $i,j$, $\beta \in $ strategies of $j$:
$$
\frac{1}{\lambda_{ij}} * \sum_{\alpha \textrm{ best response to $\beta$}} x_{ij\alpha\beta} = \sum_{k, \gamma} x_{jk \beta \gamma}
$$
\end{proof}
}
\begin{theorem}It it \textbf{PPAD}-hard to find a personalized equilibria in general matrix games.
\end{theorem}

\begin{proof}
Again, this is shown via the reduction from \gamenamePlural{ }to fractional BGP to personalized equilibria.
\end{proof}
\junk{
\subsubsection*{And gadget} 
Given 2 values $v_1, v_2 \in \{0,1\}$, we can
create a game with 2 strategies per player (0 and 1) such that in any
personalized equilibria in which player A plays 1 with fraction $v_1$
and player B plays 1 with fraction $v_2$, player C plays 1 with
fraction $v_1 \land v_2$.

$P_C(1,1,1) = P_C(0,0,0) = P_C(0,1,0) = P_C(1,0,0) = 1$. The rest of the payoffs are 0.

\subsubsection*{Or gadget}
Given 2 values $v_1, v_2 \in \{0,1\}$, we can create a game with 2
strategies per player (0 and 1) such that in any personalized
equilibria in which player A plays 1 with fraction $v_1$ and player B
plays 1 with fraction $v_2$, player C plays 1 with fraction $v_1 \lor
v_2$.

$P_C(1,1,1) = P_C(1,0,1) = P_C(0,1,1) = P_C(0,0,0) = 1$. The rest of the payoffs are 0.

\subsubsection*{Not gadget} Given value $v_1 \in \{0,1\}$, we can
create a game with 2 strategies per player (0 and 1) such that in any
personalized equilibria in which player A plays 1 with fraction $v_1$,
player B plays 1 with fraction $\lnot v_1$.

$P_B(0,1) = P_B(1,0) = 1. P_B(0,0) = P_B(1,1) = 0$.

\subsubsection*{Sum gadget} Given 2 values $v_1, v_2 \in [0,1]$, we
can create a game with 2 strategies per player (0 and 1) such that in
any personalized equilibria in which player A plays 1 with fraction
$v_1$ and player B plays 1 with fraction $v_2$, player C plays 1
max(1,$v_1 + v_2$).

$P_C(0,1,1) = P_C(1,0,1) = P_C(1,1,1) = 10. P_C(0,0,0) = 1$. The rest of the payoffs are 0. 

\subsubsection*{Difference gadget} Given 2 values $v_1, v_2 \in
[0,1]$, we can create a game with 2 strategies per player (0 and 1)
such that in any personalized equilibria in which player A plays 1
with fraction $v_1$ and player B plays 1 with fraction $v_2$, player C
plays 1 with fraction $v_1 - v_2$ if $v_1 > v_2$, or 0 otherwise.

$P_C(0,0,0)=10, P_C(1,1,0)=10, P_C(1,0,1)=1, P_C(0,1,0)=1$. The rest of the payoffs are 0.

\subsubsection*{Less than gadget} Given 2 values $v_1, v_2 \in [0,1]$,
we can create a game with 2 strategies per player (0 and 1) such that
in any personalized equilibria in which player A plays 1 with fraction
$v_1$ and player B plays 1 with fraction $v_2$, player C plays 1 if
$v_1 - v_2 \geq \epsilon$, 0 if $v_1 \leq v_2$.

First create a minus gadget $minus$ (player A - player B). Feed the
result into a copy gadget $copy_1$, and feed the result of $minus$ and
$copy_1$ into a sum gadget, $sum_1$. For $i = 1$ to $-\log \epsilon$,
feed $sum_i$ to a copy gadget $copy_i$, and feed $sum_i$ and $copy_i$
into a sum gadget $sum_{i+1}$.  Call the last sum player ``Player C''.

If $v_1 \leq v_2$, the minus gadget will return 0, so player C will
play the result of copying 0 and adding it to 0 (many times), or 0. If
$v_1 - v_2 \geq \epsilon$, player C will play the max of 1 or $v_1 -
v_2$ copied and added to itself $- \log \epsilon$ times, or $(v_1 -
v_2) * 2^{- \log \epsilon} = (v_1 - v_2) * \frac{1}{\epsilon} \geq
\frac{\epsilon}{\epsilon}$.

\subsubsection*{Half gadget}Given value $v_1 \in [0,1]$, we can
create a game with 2 strategies per player (0 and 1) such that in any
personalized equilibria in which player A plays 1 with fraction $v_1$,
player B plays 1 with fraction $\frac{1}{2} * v_1$.

Create an additional player, player C. The payoffs to players B and C are as follows:

$P_B(0,0,0)=10, P_B(0,0,1)=0, P_B(0,1,0)=1, P_B(0,1,1)=0, P_B(1,0,0)=1, P_B(1,0,1)=0, P_B(1,1,0)=0, P_B(1,1,1)=1$ \\
$P_C(0,0,0)=10, P_C(0,0,1)=1, P_C(0,1,0)=0, P_C(0,1,1)=0, P_C(1,0,0)=0, P_C(1,0,1)=1, P_C(1,1,0)=1, P_C(1,1,1)=0$ \\

As much as player A plays 0 ($1-v_1$), player B will also play
0. Player B has extra incentive to play hyperedge $\{0,0,0\}$, so it
will use up $1-v_1$ of the amount that player C plays 0 on this
hyperedge. Similarly, player C will play $1-v_1$ on hyperedge
$\{0,0,0\}$. Now, suppose player B plays 0 with fraction $1-v_i +
\alpha$ and plays 1 with fraction $\beta$. Then, we know that player C
will play $1-v_1$ on hyperedge $\{0,0,0\}$, and his payoffs against
player A's 1 show that he will play 1 against the remaining amount
that player B played 0 ($\alpha$), and 0 against the amount that
player B played 1 ($\beta$). So player C plays 0 with amount $1-v_1 +
\beta$ and plays 1 with amount $\alpha$. However, we can see that
player B will try to mimic the amounts played by player C (except for
the $1-v_1$ played against player 1's 0). So this will only be an
equilibrium if $\alpha = \beta$ and $1-v_1 + \alpha + \beta = 1$, or
if $\alpha = \beta = v_1/2$.
}
\junk{ LJP - COMMENTING THIS OUT BECAUSE IT'S SO SIMILAR TO THE D.G.P. VERSION. SHOULD BE OBVIOUS
\subsection{Algorithm for extracting digits}
From Daskalakis, Goldberg, Papadimitriou (unchanged).  Notice: $<(a,b)$ means 1 if $a < b-\epsilon$, 0 if $a \geq b$, undefined if $b-\epsilon < a \leq b$. In our gadgets above, $-, =,$ and $*2^{-i}$ are calculated exactly.
\begin{algorithm}[htb!]
\begin{algorithmic}[1]
\STATE $x_1=x$ (short for $p[v_{x_1}] = p[v_x]$):
\FOR {$i=1, \ldots, n$}
  \STATE $b_i(x) :=<(2^{-i}, x_i)$
  \STATE $x_{i+1} := x_i - (b_i(x) * 2^{-i})$
\ENDFOR
\STATE Similarly for $y$ abd $z$
\end{algorithmic}
\end{algorithm}

\begin{lemma}(Our version of lemma 2 from the Daskalakis, Goldberg, Papadimitriou paper.) For $m \leq n$, if $\sum_{i=1}^m b_i 2^{-i} + \epsilon < x < \sum_{i=1}^m b_i 2^{-i} + 2^{-m} - \epsilon$, for $b_1, \ldots, b_m \in \{0,1\}$, then at any personalized equilibrium of $G$, $p[v_{b_m}(x)] = b_m$, and $p[v_{x_{m+1}}] = x - \sum_{i=1}^m b_i 2^{-i}$.
\end{lemma}

\begin{proof}
Call $\sum_{i=1}^m b_i 2^{-i} + \epsilon < x$ \emph{Condition 1}, and call $x < \sum_{i=1}^m b_i 2^{-i} + 2^{-m} - \epsilon$ \emph{Condition 2}.

Since our subtraction and $*2^{-i}$ gadgets calculate exactly, it is
clear that if $p[v_{b_m}(x)]$ is set correctly then $p[v_{x_{m+1}}]$
will also be set correctly. Therefore, we only need to prove that
$p[v_{b_m}(x)]$ will be set correctly whenever Conditions 1 and 2
hold.

Induction on $m$. For $m=1$, if $b_1=0$, then Conditions 1 and 2 imply
that $x$ is between $\epsilon$ and $\frac{1}{2} -
\epsilon$. Therefore, the less than gadget will return 0, so
$p[v_{b_1}(x)]$ will be set correctly to 0. If $b_1=1$, then
Conditions 1 and 2 imply that $x$ is between $\frac{1}{2} + \epsilon$
and $1 - \epsilon$. Again, this means that the less than gadget will
correctly return 1, so $p[v_{b_1}(x)]$ will be set correctly to 1.

Suppose the lemma holds up through $m$. $\sum_{i=1}^{m} b_i 2^{-i} +
\epsilon \geq \sum_{i=1}^{m-1} b_i 2^{-i} + \epsilon$, so if Condition
1 holds at $m$, then Condition 1 held at $m-1$ as well.  $b_m * 2^{-m}
\leq 2^{-m}$, so $\sum_{i=1}^{m} b_i + 2^{-m} \leq \sum_{i=1}^{m-1}
b_i + 2^{-m} + 2^{-m} = \sum_{i=1}^{m-1} b_i + 2^{-(m-1)}$. So if
Condition 2 holds at $m$, then Condition 2 held at $m-1$ as
well. Therefore, if Conditions 1 and 2 hold at $m$, then by the
inductive assumption, $p[v_{b_{(m-1)}}(x)] = b_{(m-1)}$, and
$p[v_{x_m}] = x - \sum_{i=1}^{m-1} b_i 2^{-i}$.  The algorithm sets
$p[v_{b_m}(x)] :=<(2^{-m}, p[v_{x_m}])$.  Condition 1 says that $x -
\sum_{i=1}^m b_i 2^{-i} > \epsilon$, or $p[v_{x_m}] - b_m 2^{-m} >
\epsilon$, so if $b_m = 1$, the less than gadget will work correctly
(setting $p[v_{b_m}(x)]$ to 1). Condition 2 says that $\epsilon <
\sum_{i=1}^m b_i 2^{-i} + 2^{-m} - x$, or $\epsilon < - p[v_{x_m}] +
b_m 2^{-m} + 2^{-m}$, and if $b_m = 0$, $p[v_{x_m}] < p[v_{x_m}] +
\epsilon < 2^{-m}$, so the less than gadget will correctly set
$p[v_{b_m}(x)]$ to 0.

\end{proof}
}

\begin{theorem} \label{thm:personalizedPPAD}
It is \textbf{PPAD}-hard to find 4-player personalized equilibria.
\end{theorem}

\begin{proof}
For this theorem, we first note that when reducing from \gamenamePlural{ }to fractional BGP to personalized equilibria, we keep the same number of players. We also preserve the number of players ``depended on'' for the payoff of a particular strategy. In \gamenamePlural, a payoff depends only on the single other player being chosen, and a player will only place weight on other players it prefers over itself. When reduced to fractional BGP, the paths considered by a node are only one or two-hop paths: and the node only considers two-hop paths that it prefers over its ``direct'' path. When reduced to finding personalized equilibria in a matrix game, we keep the same number of players. The payoff for a hyperedge depends only on a number of players equal to the number of hops in the path represented by that hyperedge; in this case, at most 2. The only strategies that will ever be chosen by a node are the strategies corresponding to players prefered over that node in the original instance of the\gamename.

\junk{
This is essentially the same reduction as in the Goldberg et al paper, but
it's actually a little easier for us because of the personalization.
}

Therefore, if we start with the reduction from 3-DIMENSIONAL BROUWER used in Section \ref{sec:ppad-hardness}, we have a graphical matrix game in which each node's strategy ``depends on'' one or two other nodes, and each node ``influences'' the strategy for one or two other nodes.
(We can easily make this a max of 2, because if a
node influences 3 other nodes, we just switch to influencing one of
them plus a copy gadget, which can influence the other(s). So we can
say that in+out degree is at most 4. Each node has 2 strategies.

We want to represent this as a 4 player game with more strategies per
player.

First we have to slightly transform the graph in order to get two
properties:
\begin{enumerate}
\item Our graph should have max degree of 3 (in + out)
\item If a node ``depends on'' 2 other nodes, we want an edge
(undirected is fine) between these two nodes. This edge counts in the
degree.
\end{enumerate}

We have a degree 4 graph and we want to change it to degree
3. Any degree 4 node currently has 2 inputs plus 2 outputs. We can
change both of the outputs into a single copy gadget with 2 outputs -
the copy has one input and 2 outputs, or degree three.

Now we have a graph with max degree 3, but we need the edges from
property 2 (an edge between any two nodes $X$ and $Y$ that influence
the same third node, $Z$). Suppose we have two nodes $X$ and $Y$ that
both influence a third node $Z$, and node $X$ has degree 3
already. Just add a new node $X'$ that copies $X$, and make this copy
influence $Z$. Now $X$ still has degree 3, $X'$ has degree 2 + the
edge between $X'$ and $Y$.

After the above conversions, we have a graph with those 2 properties. Create a 3-coloring
(possible because max degree is 3). Create one player per color. Each
player has 2 strategies for each node in that color (one for the 0
strategy of that node, one for the 1 strategy).

Add dummy strategies as necessary so that each of the 3 players has
the same number of strategies. Also add a fourth player with half the
number of strategies as any other player.

This gives us 4 players. Let the strategies for player 1 be $\{a_{10},
a_{11}, a_{20}, a_{21}, \ldots, a_{k0}, a_{k1}\}$. The strategies for
player 2 are $\{b_{10}, b_{11}, \ldots, b_{k0}, b_{k1}\}$. The
strategies for player 3 are $\{c_{10}, c_{11}, \ldots, c_{k0},
c_{k1}\}$. The strategies for player 4 are $\{d_1, d_2, \ldots,
d_k\}$.

Next we will assign payoffs for each hyperedge. Start by giving each hyperedge the same payoff as in the graphical
game (we can do this because no two nodes influencing the same
strategy are strategies of the same player). Notice that these payoffs
will not depend at all on player 4. All of player 4's payoffs start at
0. Let $p_i(w,x,y,z)$ = the payoff to player $i$ if player 1 plays
$w$, 2 plays $x$, player 3 plays $y$, player 4 plays $z$. Now we want
to add to these payoffs in order to ensure that each player plays each
strategy pair equally.

Let $M >$ the largest payoff so far. Now, change the following payoffs:

$p_1(a_{si}, x, y, d_{s}) += M$ (player 1 is playing either strategy from the node numbered $s$, player 4 is playing his $s^{th}$ strategy). \\
$p_2(w, b_{si}, y, d_{s}) += M$ (player 2 is playing either strategy from the node numbered $s$, player 4 is playing his $s^{th}$ strategy). \\
$p_3(w, x, c_{si}, d_{s}) += M$ (player 3 is playing either strategy from the node numbered $s$, player 4 is playing his $s^{th}$ strategy). \\
$p_4(a_{si}, x, y, d_{(s+1)}) += M$ (player 1 is playing either strategy from the node numbered $s$, player 4 is playing his $s+1^{st}$ strategy). \\

If $f_i(x)$ = the amount player $i$ plays strategy $x$ then in any
equilibrium we must have (for all $s$)
\begin{eqnarray*}
f_1(a_{s0}) + f_1(a_{s1}) & = & f_4(d_s) \\
f_2(b_{s0}) + f_2(b_{s1}) & = & f_4(d_s) \\
f_3(c_{s0}) + f_3(c_{s1}) & = & f_4(d_s) \\
f_4(d_s)  & = & f_1(a_{(s-1)0}) + f_1(a_{(s-1)1}) \textrm{for $s > 0$} \\
f_4(d_0)  & = & f_1(a_{k0}) + f_1(a_{k1})\\
\end{eqnarray*}

These equations imply that:
\begin{eqnarray*}
f_1(a_{s0}) + f_1(a_{s1}) & = & f_1(a_{(s-1)0}) + f_1(a_{(s-1)1})  \textrm{for $s > 0$} \\
f_1(a_{00}) + f_1(a_{01}) & = & f_1(a_{k0}) + f_1(a_{k1})  \textrm{for $s > 0$} \\
f_2(b_{s0}) + f_2(b_{s1}) & = & f_1(a_{s0}) + f_1(a_{s1}) \\
f_3(c_{s0}) + f_3(c_{s1}) & = & f_1(a_{s0}) + f_1(a_{s1}) \\
\end{eqnarray*}

In other words, given an equilibrium in this game, we can simply
multiply by the number of pairs (nodes) per player to get an
equilibrium in the graphical game.
\end{proof}

With a slight modification, the
above proof also establishes \textbf{PPAD}-hardness for approximating
personalized equilbria in 5-person games. We simply need to note that in the approximation gadgets created for the \gamename\ hardness proof, the maximum length of any preference list is 4 instead of 3 (one of the players in the CORRECTION gadget ``depends on'' 3 other players). The rest of the proof remains intact.

\section{Concluding Remarks}
We note that our \textbf{PPAD}-hardness results from section \ref{sec:hardness} also apply to two other problems reduced in \cite{HaxellWilfongJournal} to fractional BGP. The first of these problems is finding fractional stable matchings in hypergraphic preference systems. 
\junk{Here, we are given a hypergraph plus a preference list for each node. A fractional solution is one in which every hypergraph edge is assigned a weight such that the total weight adjacent to each node is at most 1. A fractional stable solution also ensures that for each edge $E$ in the hypergraph, some node in the edge has a total of weight 1 assigned to adjacent edges that it prefers at least as much as $E$. }
The second is finding fractional kernels in directed graphs.

We raise a number of open questions.
\begin{itemize}
\item Is finding a personalized equilibrium in general matrix games in \textbf{PPAD}? Although we show \textbf{PPAD}-hardness for general games, we have not settled the question of \textbf{PPAD}-membership. We have shown that a rational solution always exists, so finding an exact equilibrium may be in \textbf{PPAD}.
\item We show that it is possible to find a personalized equilibrium for a 2-person game in polynomial time, and it is \textbf{PPAD}-hard to find a personalized equilibrium in a 4-person game. However, the hardness of finding these equilibria in 3-person games remains open.
\item In this paper, we concentrate on a version of BBC games in which all nodes want to reach a single universal destination. However, in \cite{OurPODC08}, the utility of a node in a BBC game is defined as an affinity-weighted average of the shortest path length (or minimum cost flow in the fractional case) to \emph{all} other nodes. They show that an equilibrium always exists with multiple destinations, and our hardness results of course extend to this model, since we can define only one non-zero affinity, but it is unknown whether rational equilibria always exist.
\item Our reduction from \gamenamePlural\ to BBC games does not apply in the ``multiple destinations'' model if all affinities must be equal. Is this special instance of BBC games \textbf{PPAD}-hard as well?
\end{itemize}

Personalized equilibria and \gamenamePlural\ both have a number of real world and theoretical applications, and seem to be natural end points in a spectrum of ``personalized'' fractional games. As more games are added to this hierarchy, we hope to fully understand the relationships and behavior of fractional equilibria.


\begin{thebibliography}{99}

\bibitem{AlbersEiltsEvenDarMansourRoditty}
Susanne Albers, Stefan Eilts, Eyal Even-Dar, Yishay Mansour, and Liam Roditty.
\newblock On {N}ash equilibria for a network creation game.
\newblock In {\em Proc. of SODA '06}, pages 89--98, New York, NY, USA, 2006.
  ACM Press.

\bibitem{AnshelevichShepherdWilfong}
Elliot Anshelevich, Bruce Shepherd, and Gordon Wilfong.
\newblock Strategic network formation through peering and service agreements.
\newblock In {\em Proc. of IEEE FOCS '06}, 77--86, Washington, DC, USA,
  2006.

\bibitem{Aumann}
R.J. Aumann.
\newblock Subjectivity and Correlation in Randomized Strategies.
\newblock \emph{Journal of Mathematical Economics}, 1:67-96, 1974.

\bibitem{BalaGoyal}
Venkatesh Bala and Sanjeev Goyal.
\newblock A noncooperative model of network formation.
\newblock {\em Econometrica}, 68(5):1181--1229, 2000.

\bibitem{ChenDengTeng}
X.~Chen, X.~Deng, and S.-H.~Teng.
\newblock Computing Nash Equilibria: Approximation and Smoothed Complexity.
\newblock In {\em FOCS}, 603--612, 2006.

\bibitem{ChenDengTengJACM}
X.~Chen, X.~Deng, and S.-H. Teng.
\newblock Settling the complexity of computing two-player {Nash} equilibria.
\newblock {\em JACM}, (invited and under review), 2008.

\bibitem{DaskalakisGoldbergPapadimitriou}
Constantinos Daskalakis, Paul W. Goldberg, Christos H, Papadimitriou.
\newblock The Complexity of Computing a {N}ash Equilibrium
\newblock In {\em STOC '06}, 2006.

\bibitem{DemaineHajiaghaviMahini}
Erik D. Demaine, MohammadTaghi Hajiaghavi, and Hamid Mahini.
\newblock The Price of Anarchy in Network Creation Games
\newblock In {\em PODC}, pages 292-298, 2007.

\bibitem{EvenDarKearns}
Eyal Even-Dar and Michael Kearns.
\newblock A small world threshold for economic network formation.
\newblock In {\em NIPS}, pages 385--392, 2006.

\bibitem{FabrikantLuthraManevaPapadimitriouShenker}
Alex Fabrikant, Ankur Luthra, Elitza Maneva, Christos~H. Papadimitriou, and
  Scott Shenker.
\newblock On a network creation game.
\newblock In {\em PODC '03}, pages 347--351, New York, NY, USA, 2003. ACM
  Press.

\bibitem{FeigenbaumPapadimitriouSamiShenker}
J.~Feigenbaum, C.~Papadimitriou, R.~Sami, and S.~Shenker.
\newblock A {BGP}-based mechanism for lowest-cost routing.
\newblock In {\em PODC}, 2002.

\bibitem{GareyJohnson}
M.~R.~Garey, and D.~S.~Johnson.
\newblock Computers and intractability.
\newblock {\em Freeman Press}, 1979.

\junk{
\bibitem{GovindanReddy}
R.~Govindan, and A.~Reddy.
\newblock An analysis of inter-domain topology and route stability.
\newblock In {\em INFOCOM}, 1997.
}

\bibitem{GriffinShepherdWilfong}
Timothy G. Griffin, F. Bruce Shepherd, and Gordon Wilfong.
\newblock The stable paths problem and interdomain routing.
\newblock {\em IEEE/ACM Transactions on Networking}, 2002.

\bibitem{HaleviMansour}
Yair Halevi and Yishay Mansour.
\newblock A Network Creation Game with Nonuniform Interests.
\newblock In {\em WINE}, pages 278-292, 2007.

\bibitem{HaxellWilfong}
P. E. Haxell and G. T. Wilfong.
\newblock A fractional model of the border gateway protocol ({BGP}).
\newblock In {\em SODA}, pages 193-1999, 2008.

\bibitem{HaxellWilfongJournal}
P. E. Haxell and G. T. Wilfong.
\newblock On the Stable Paths Problem.
\newblock Preprint, 2008.

\bibitem{JacksonWolinsky}
Matthew Jackson and Asher Wolinsky.
\newblock A strategic model of social and economic networks.
\newblock {\em Journal of Economic Theory}, 71:44--74, 1996.

\bibitem{JohariMannorTsitsiklis}
Ramesh Johari, Shie Mannor, and John~N. Tsitsiklis.
\newblock A contract-based model for directed network formation.
\newblock {\em Games and Economic Behavior}, 56(2):201--224, 2006.

\bibitem{KearnsLittmanSingh}
M. Kearns, M.L. Littman, S. Singh.
\newblock Graphical models for game theory.
\newblock \emph{UAI}, 253-260, 2001.

\bibitem{Kintali}
Shiva Kintali.
\newblock A Distributed Protocol for Fractional Stable Paths Problem.
\newblock http://www.cc.gatech.edu/research/reports/GT-CS-08-06.pdf

\bibitem{KoutsoupiasPapadimitriou}
Elias Koutsoupias and Christos Papadimitriou.
\newblock Worst-case equilibria.
\newblock In {\em STACS '99}, pages 404--413, March 1999.

\junk{
\bibitem{LabovitzMalanJahanian}
C.~Labovitz, G.~R.~Malan, and F.~Jahanian.
\newblock Internet routing instability.
\newblock In {\em SIGCOMM}, 1997.
}

\bibitem{OurPODC08} 
Nikolaos Laoutaris, Laura J. Poplawski, Rajmohan Rajaraman, Ravi Sundaram, Shang-Hua Teng.
\newblock Bounded Budget Connection (BBC) Games or How to Make Friends and Influence People, on a Budget.
\newblock In {\em PODC '08}, pages 165--174, 2008.

\bibitem{LaoutarisPRST} 
Nikolaos Laoutaris, Laura J. Poplawski, Rajmohan Rajaraman, Ravi Sundaram, Shang-Hua Teng.
\newblock Bounded Budget Connection (BBC) Games or How to make friends and influence people, on a budget.
\newblock arXiv:0806.1727v1 [cs.GT]

\bibitem{MarkakisSaberi}
Evangelos Markakis and Amin Saberi.
\newblock On the core of the multicommodity flow game.
\newblock In {\em Proc. of the 4th ACM conference on Electronic commerce},
  pages 93--97, New York, NY, USA, 2003. ACM Press.

\bibitem{Nash50}
J.~Nash.
\newblock Equilibrium point in n-person games.
\newblock In {\em PNAS}, 36(1):48--49, 1950.

\bibitem{Nash51}
J.~Nash.
\newblock Noncooperative games.
\newblock In {\em Annals of Mathematics} 54:286--295, 1951.

\bibitem{NisanRoughgardenTardosVazirani}
N.~Nisan, T.~Roughgarden, E.~Tardos, and V.~Vazirani.
\newblock Algorithmic game theory.
\newblock {\em Cambridge University Press}, 2007.

\bibitem{OsborneRubinstein}
M.J. Osborne and A.~Rubinstein.
\newblock {\em A Course in Game Theory}.
\newblock MIT Press, 1994.

\bibitem{Papadimitriou94}
C.~Papadimitriou.
\newblock On the Complexity of the Parity Argument and Other Inefficient Proofs of Existence.
\newblock {\em JCSS} 48(3):498--532, 1994.

\bibitem{Papadimitriou01}
C. Papadimitriou.
\newblock Algorithms, games, and the {I}nternet.
\newblock In {\em STOC '01}, 749--753, New York, NY, USA, 2001. ACM
  Press.

\bibitem{RehkterLi}
Y. Rehkter, T. Li.
\newblock A Border Gateway Protocol (BGP version 4).
\newblock RFC 1771, 1995.

\bibitem{Stewart}
J.~W~Stewart.
\newblock BGP4: Inter-domain routing in the {I}nternet.
\newblock Addison Wesley, 1998. 

\bibitem{VaradhanGovindanEstrin}
K. Varadhan, R. Govindan, and D. Estrin.
\newblock Persitent Route Oscillations in Inter-Domain Routing.
\newblock Technical Report USC CS TR 96-631, Department of Computer Science, University of Southern California, Feb. 1996.

\bibitem{Yanovskaya}
E.~B.~Yanovskaya.
\newblock Equilibrium situations in multi-matrix games.
\newblock {\em Litovskii Matematicheskii Sbornik}, 8:381--384, 1968.

\end{thebibliography}
\end{document}